\documentclass[letterpaper,twocolumn,10pt]{article}
\usepackage{usenix2019,epsfig,endnotes}
\usepackage[symbol]{footmisc}

\usepackage{amsfonts}
\usepackage{xspace}
\usepackage{pifont}
\usepackage{multirow}
\usepackage[english]{babel}
\usepackage{graphicx}
\usepackage{url}
\usepackage[T1]{fontenc}
\usepackage[T1]{fontenc}       
\usepackage[latin1]{inputenc}  
\usepackage[english]{babel}    
\usepackage{lmodern}
\usepackage{fancyvrb}
\usepackage{pgfplots}
\usepackage{mathptmx}
\usepackage{amsmath,amsthm,stmaryrd,amsfonts,mathtools}
\usepackage{algorithm,algpseudocode}
\usepackage{amssymb}

\usepackage{wrapfig}
\usepackage{graphicx}
\usepackage{caption}
\usepackage{appendix}
\usepackage{fancybox}
\usepackage{calc}
\usepackage{url}
\usepackage{booktabs}
\usepackage{times}
\usepackage{enumitem}
\usepackage{subcaption}

\newenvironment{packeditemize}{\begin{list}{$\bullet$}{\setlength{\itemsep}{0pt}\addtolength{\labelwidth}{-5pt}\setlength{\leftmargin}{\labelwidth}\setlength{\listparindent}{\parindent}\setlength{\parsep}{0pt}\setlength{\topsep}{3pt}}}{\end{list}}

\newcommand{\cut}[1]{}

\newcommand{\eps}{\varepsilon}

\newtheorem{theorem}{Theorem}
\newtheorem{corollary}{Corollary}
\newtheorem{lemma}{Lemma}

\begin{document}

\title{Plankton: Scalable network configuration verification through model checking\vspace{-5cm}}

\maketitle
\vskip -2cm
\thispagestyle{empty}
\begin{abstract}
A variety of network verification tools have been proposed in the recent
years, using different network abstractions and analysis techniques. For configuration verification, techniques ranging from graph algorithms to SMT solvers have been proposed. In spite of this, scalable configuration verification with sufficient protocol support continues to be a challenge.  In this paper, we show that by combining equivalence partitioning
with explicit-state model checking, network configuration verification can be scaled
significantly better than the state of the art, while still supporting a rich set of protocol features. We propose Plankton, which uses symbolic partitioning to manage large header spaces and efficient model checking to exhaustively explore protocol behavior. Thanks to a highly effective suite of optimizations including state hashing, Partial Order Reduction, and Policy-based Pruning, Plankton successfully verifies policies in industrial-scale networks quickly and compactly, at times reaching a 10000$\times$ speedup compared to the state of the art.
\end{abstract}
 \vspace{-0.3cm}
\section{Introduction}
\label{sec:intro}

Ensuring correctness of networks is a difficult and critical task. A growing number of network verification tools are targeted towards automating this
process as much as possible, thereby reducing the burden on the network operator. Verification platforms have improved steadily in the recent years,
both in terms of scope and scale. Starting from offline data plane verification tools like Anteater~\cite{cite:anteater} and HSA~\cite{cite:hsa},
the state of the art has evolved to support real-time data plane verification~\cite{cite:veriflow, cite:netplumber}, and more recently, analysis
of configurations~\cite{cite:batfish, cite:era, cite:arc, cite:minesweeper}.

Configuration analysis tools such as Batfish~\cite{cite:batfish}, ERA\cite{cite:era}, ARC\cite{cite:arc} and Minesweeper \cite{cite:minesweeper} are designed to take as input a given network configuration, a correctness specification, and possibly an ``environment'' specification, such as the maximum expected number of failures, external route advertisements etc. Their task is to determine whether, under the given environment specification, the network configuration will always meet the correctness specification. As with most formal verification domains, the biggest challenge in configuration analysis is that of scalability. Being able to analyze the behavior of multiple protocols executing together, while also accounting for the large space of packet headers, is a non-trivial task. Past verifiers have used a variety of techniques to try and surmount this scalability challenge. While many of them sacrifice their correctness or expressiveness in the process, Minesweeper maintains both by modeling the network using SMT constraints and performing the verification using an SMT solver. However, we observe that this approach scales poorly with increasing problem size (4$+$ hours to check a 245 device network for loops, in our experiments). So, in this paper, we attempt to answer the following question: {\it Can we build a feature-rich configuration verification tool that can scale well to practical verification challenges?}

We begin our work by observing that scalability challenges in configuration verification stem from two factors --- the large space of possible packet headers, and the possibly diverse outcomes of control plane execution, particularly in the presence of failures. We believe that even with efficient data structures and symbolic search techniques, SAT/SMT solvers cannot tackle these challenges as effectively as techniques specifically designed to address them. In fact, these challenges have been studied extensively, in the domains of data plane verification and software verification. Data plane verification tools analyze the large header space to determine all possible data plane behaviors and check their correctness. Software verification techniques explore the execution paths of software, including distributed software, and identify undesirable executions that often elude testing. Motivated by the success of the analysis algorithms in these domains, we attempted to combine the two into a scalable configuration verification platform. The result --- Plankton --- is a configuration verifier that uses equivalence partitioning to manage the large header space, and explicit-state model checking to explore protocol execution. Thanks to these efficient analysis techniques, and an extensive suite of domain-specific optimizations, Plankton delivers consistently high verification performance.

Our key contributions are as follows:
\begin{packeditemize}
\item We define a configuration verification paradigm that combines packet equivalence computation and explicit-state model checking.
\item We develop optimizations that make the design feasible for networks of practical scale, including optimizations to reduce the space of event exploration,
and techniques to improve efficiency of exploration.
\item We implement a Plankton prototype with support for OSPF, BGP and static routing, and show experimentally that it can verify policies at practical scale (less than a second for networks with over 300 devices). Plankton outperforms the state of the art in all our tests, in some cases by as much as 4 orders of magnitude.
\end{packeditemize}
 \vspace{-0.62cm}
\section{Motivation and Design Principles}
\vspace{-0.18cm}
\label{sec:motivation}

\begin{figure}
    \centering
   \scalebox{0.56}{
    \begin{tabular}{ p{0.52\columnwidth} | p{0.09\columnwidth} | p{0.11\columnwidth}| p{0.1\columnwidth}| p{0.3\columnwidth}| p{0.2\columnwidth}| p{0.1\columnwidth} }
      \hline
      {\small Feature} & {\small Batfish} & {\small BagPipe} & {\small ARC} & {\small ERA} & {\small Minesweeper} & {\small Plankton}\\
      \hline
      \hline
      {\small All data planes, including failures} & {\small No} & {\small No }& {\small Partial} & {\small No} & {\small Yes} & {\small Yes}\\
      \hline
      {\small  Support beyond specific protocols} & {\small Yes} & {\small No }& {\small No} & {\small Yes} & {\small Yes} & {\small Yes}\\
      \hline
      {\small Soundness when assumptions hold} & {\small Yes} & {\small Yes} & {\small Yes} & {\small Only Segmentation} & {\small Yes} & {\small Yes}\\
      \hline
    \end{tabular}
    }
    \caption{Comparison of configuration verification systems}
    \label{stateoftheart}
\end{figure}
Configuration verification systems are designed to take as input the network configuration files, and combine them with an abstraction of the lower layers --- the distributed control plane, which executes to produce data plane state, and the forwarding hardware that will use that state to process traffic. They may additionally take an {\em environment specification}, which describes interactions of entities external to the network, such as possible link failures, route advertisements etc. Their task is to verify that under executions enabled by the supplied configuration and environmental conditions, the correctness requirements are never violated. Configuration verifiers thus help ensure correctness of proposed configurations under all possible network conditions, prior to deployment. Figure~\ref{stateoftheart} illustrates the current state of the art in configuration verification. It can be seen that Minesweeper is the only existing tool that can reason about multiple possible converged data plane states of the network, including those produced by topology changes, while also having the ability to support more than just a specific protocol, and maintaining soundness of analysis. All other tools compromise on one or more key functionality. ARC for example uses graph algorithms to compute the multiple converged states enabled by failures, but only for shortest-path routing. As a result it cannot handle common network configurations such as BGP configurations that use LocalPref, any form of recursive routing etc. The reason for the mismatch in Minesweeper's functionality in contrast to others is that it makes a different compromise. By using an SMT-based formulation, Minesweeper is able to achieve good data plane coverage and feature coverage, but pays the price in performance. As experiments show\cite{cite:cp-compression}, Minesweeper scales poorly with network size, and is unable to handle network larger than a few hundred devices in reasonable time. Our motivation for Plankton is simple --- can we design a configuration verification tool without compromising scale or functionality?

Achieving our goal requires tackling two major challenges: packet diversity and data plane diversity. Packet diversity, which refers to the large space of packets that needs to be checked, is easier to handle. Guided by past work, we leverage the notion of {\em Packet Equivalence Classes} (PECs), which are sets of packets that behave identically. Using a trie-based technique similar to VeriFlow, we compute PECs as a partitioning of the packet header space such that the behavior of all packets in a PEC remains the same throughout Plankton's exploration.

Data plane diversity refers to the complexity of checking every possible converged data plane that an input network configuration may produce, with or without external interactions such as failures. It is the task of the control plane model to determine this set of converged states, over which the supplied policy is checked. So, the choice of the control plane model determines the coverage of the verifier. Simulation based tools (the best example being Batfish) execute the system only along a single non-deterministic path, and can hence miss violations in networks that have multiple stable convergences, such as certain BGP configurations. ARC's graph-based approach accounts for failures, but can support only shortest-path routing. In order to overcome these shortcomings, Minesweeper, the current state of the art, uses an SMT solver to search through possible failures and non-deterministic protocol convergence, to find any converged state that represents a violation of network correctness. While this SMT-based solution is functionally sufficient, the generic search technique employed by SMT solvers are often too costly for the problem at hand. Network control planes operate using simple network algorithms which can not only be easily modeled in software, but also are capable of interpreting a given configuration much more quickly than SMT solving. Consider Figure~\ref{fig:smtvsdijkstra}. It illustrates the time taken to compute shortest paths in fat trees, using an SMT formulation, and using a model checker on a software implementation of shortest path routing. As the figure shows, the software model performs multiple orders of magnitude better (around $12000\times$) than the SMT-based implementation.

\begin{figure}
\begin{center}
\includegraphics[width=0.6\columnwidth, keepaspectratio]{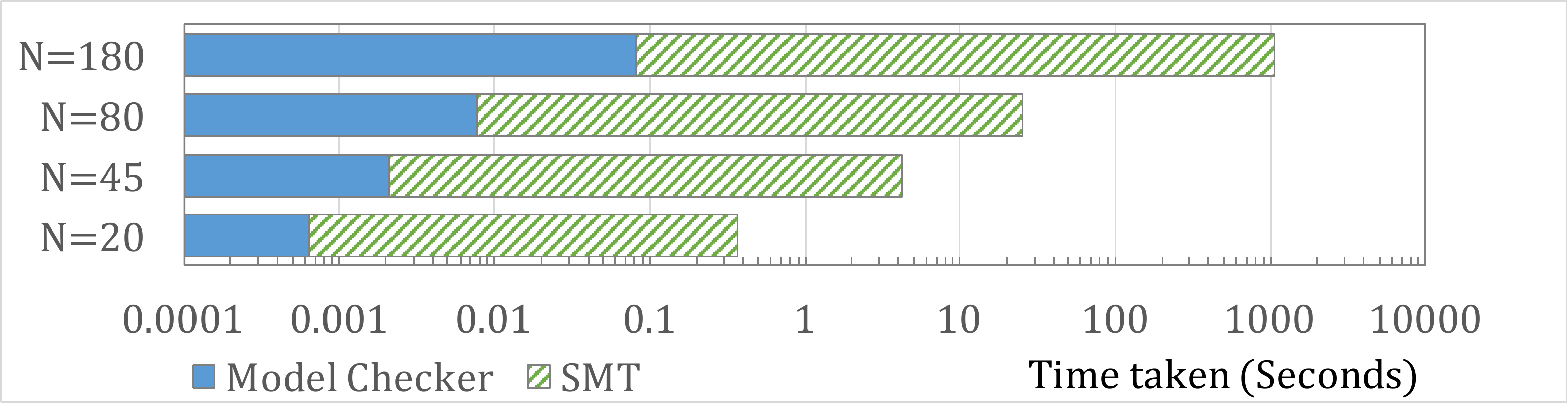}
\end{center}
\vspace{-0.7cm}
\caption{Time taken to compute shortest paths to a destination in fat trees}
\label{fig:smtvsdijkstra}
\vspace{-0.6cm}
\end{figure}

It may seem that SMT's search algorithms will pay dividends in cases where there is significant non-determinism in the network. We believe such cases are infrequent in practice. Network configurations and protocols are designed to keep the outcome of protocol execution almost fully deterministic. Non-determinism in configuration verification is usually in the choice of what failure to explore. While the protocols do have many possible non deterministic execution paths, they usually all lead to one or (in rare cases) a small number of different converged states. Motivated by this fact, create a control plane model in Promela in that incorporates the possible non-deterministic behaviors, but also implement optimizations so that when the {\em outcome} of the executions is actually deterministic, the executions get pruned to a point where the performance is comparable to simulation.  This model is exhaustively explored by SPIN, a model checker designed for Promela programs. SPIN performs a depth-first search on the state space of the program, looking for states that violate the policy being checked. We assist SPIN through optimizations that minimize the size of individual states, thus making the traversal process more efficient. Thanks to the optimizations, Plankton achieves our goal of scalable, general-purpose configuration verification.
 \vspace{-0.6cm}
\section{Plankton Design}
\vspace{-0.2cm}
\label{sec:design}
We now present Plankton's design, illustrated in Figure~\ref{fig:design}.
\begin{figure}
\begin{center}
\includegraphics[width=0.6\columnwidth, keepaspectratio]{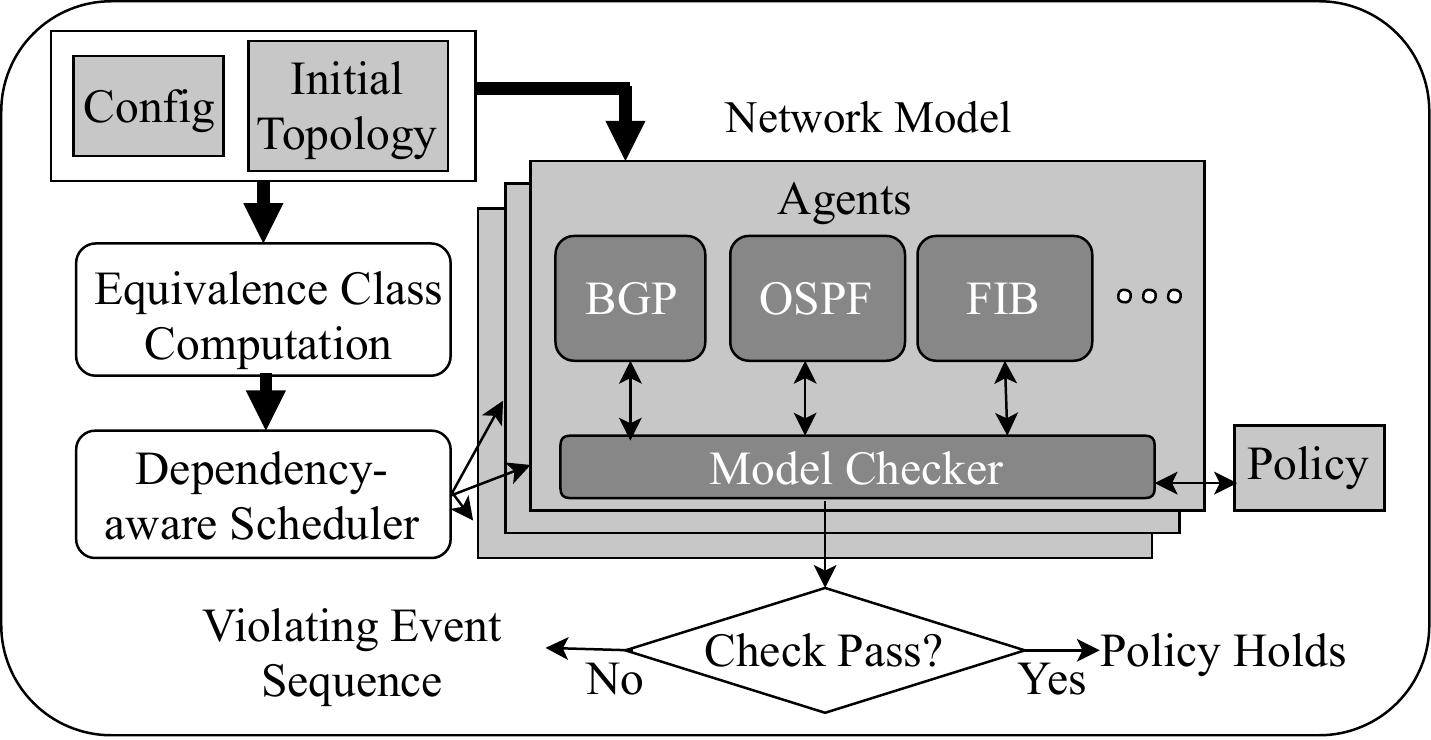}
\end{center}
\vspace{-0.5cm}
\caption{Plankton Design}
\label{fig:design}
\vspace{-0.5cm}
\end{figure}

\begin{figure*}
\centering
\begin{minipage}[b]{.75\textwidth}
\includegraphics[width=\textwidth]{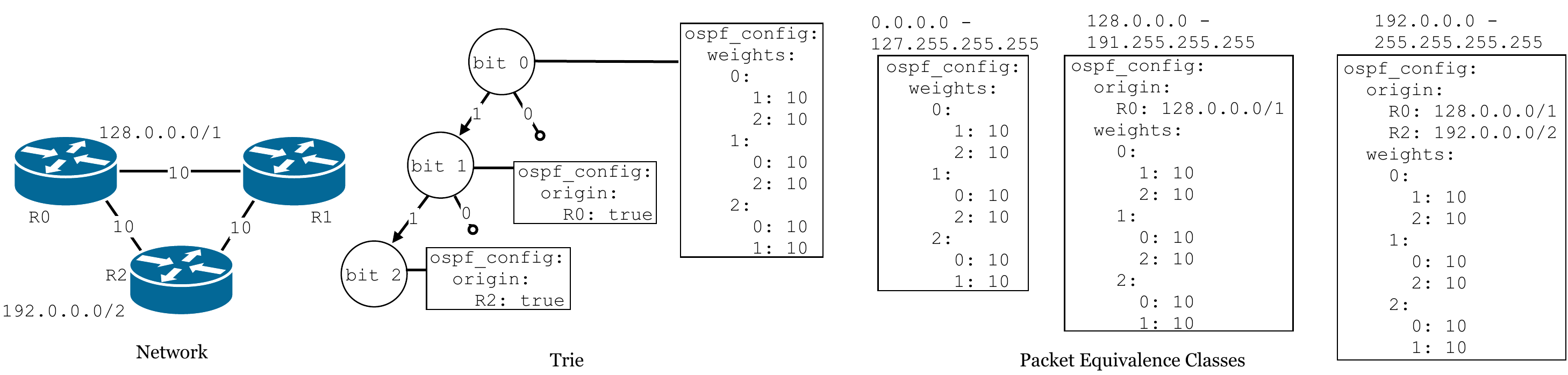}
\caption{Packet Equivalence Class computation}
\label{fig:ec-computation}
\end{minipage}\qquad
\begin{minipage}[b]{.2\textwidth}
\includegraphics[width=0.8\textwidth]{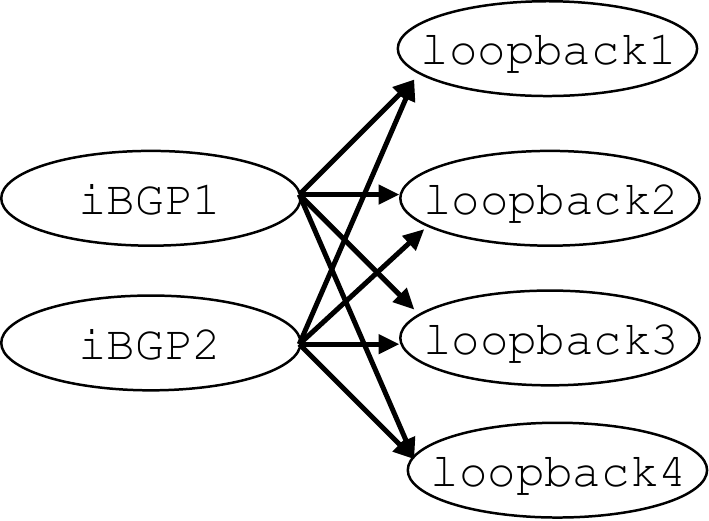}
\caption{PEC Dependency Graph}
\label{fig:ec-dependency}
\end{minipage}
\vspace{-0.25cm}
\end{figure*}

\vspace{-0.3cm}
\subsection{Packet Equivalence Classes}
\label{subsec:pec}
The first phase in Plankton's analysis is the computation of Packet Equivalence Classes. As we discussed in \S~\ref{sec:motivation}, Plankton uses a trie-based technique
inspired by VeriFlow. The trie in Plankton holds prefixes obtained from the configuration, including any prefixes that are advertised (explicitly or automatically), any prefixes appearing in route maps, any static routes etc. Each prefix in the trie is associated with a config object, that describes any configuration information that is specific to that prefix. For example, consider Figure~\ref{fig:ec-computation}, which illustrates a highly simplified example where the prefixes \texttt{128.0.0.0/1} and \texttt{192.0.0.0/2} are being advertised over OSPF in a topology with 3 devices. The trie holds three config objects --- the default, and one for each advertised prefix.

Once the trie is populated, Plankton performs a recursive traversal, simultaneously keeping track of where the prefix boundaries define division of the header space. For each known partition, it stores the most up-to-date network-wide config known. When the end of any new prefix is reached, the config object that is associated with it is merged with the network-wide config for the partition that denotes the prefix.  In our simple example, the traversal produces three classes --- \texttt{3221225472 - 4294967295}, with two nodes originating prefixes, \texttt{2147483648 - 3221225471} with only one origin, and \texttt{0 - 2147483647} without any node originating any prefix. As the example shows, each PEC-specific configuration computed this way will still include information about the original prefixes contributing to the PEC. Storing these prefixes may seem redundant. However, note that even within a PEC, the lengths of the prefixes that get advertised or get matched in route filters play a role in decision making.

\subsection{Dependency-aware Scheduling}
\label{subsec:scheduling}
It is tempting to believe that Packet Equivalence Classes could be analyzed fully independently of each other, and that an {\em embarrassingly parallel} scheme could be used to maximize parallelism in the verification process. While this is indeed true sometimes, there can often be dependencies introduced between various PECs. For example, consider a network that is running iBGP. For the various peering nodes to be able to communicate with each other, an underlying routing protocol such as OSPF should first establish a data plane state that forwards packets destined to the devices involved in BGP. In such a network, the manner in which OSPF determines the forwarding behavior for the device addresses will influence the routing decisions made in BGP. In other words, the PECs that are handled by BGP depend on those handled by OSPF. In the past, configuration verification tools have either ignored such cases altogether, or, in the case of Minesweeper, modeled these classes simultaneously. Specifically, for a network of with $n$ routers running iBGP, Minesweeper creates $n+1$ copies of the network, searching for the converged solution for the $n$ loopback addresses and also BGP. Effectively, this turns verification problem into one quadratically larger than the original problem. Given that configuration verification scales significantly worse than linearly in input size, such a quadratic increase in input size often makes the problem intractable.

Plankton goes for a more surgical approach. Once the PECs are calculated, Plankton identifies dependencies between the Packet Equivalence Classes, based on recursive routing entries, BGP sessions etc. The dependencies are stored as a directed graph, whose nodes are the PECs, and directed edges indicate which PECs depend on which others. In order to maximize parallelism in verification runs across PECs, a dependency-aware scheduler first identifies strongly connected components in the graph. These SCCs represent groups of PECs that are mutually dependent, and hence need to be analyzed simultaneously through a single verification run. In addition, if an SCC is reachable from another, it indicates that the upstream SCC can be scheduled for analysis only after the one downstream has finished. Each verification run is a separate process. For an SCC $S$, if there is another SCC $S'$ that depends on it, Plankton forces all possible outcomes of $S$ to be written to an in-memory filesystem ($S$ will always gets scheduled first). Outcomes refer to every possible converged state for $S$, together with the non-deterministic choices made in the process of arriving at them. When the verification of $S'$ gets scheduled, it reads these converged states, and uses them when necessary.

Minesweeper's technique of replicating the network roughly corresponds to the case where all PECs fall into a single strongly connected component. We expect this
to almost never be the case. In fact, in typical cases, the SCCs are likely to be of size 1, meaning that every PEC can be analyzed in isolation, with ordering of the runs being the only constraint\footnote{An example where the SCCs are bigger than one PEC is the contrived case where there exists a static route for IP A pointing to IP B, but another static route for IP B points to IP A.}.
For example, Figure~\ref{fig:ec-dependency} illustrates the dependency graph for a typical case where two PECs are being handled in iBGP on a network with 4 different routers. The only
constraint in this case is that the loopback PECs should be analyzed before the iBGP PECs can start. In such cases,
Plankton's approach not only keeps the problem size significantly smaller, it also maximizes the degree of parallelism that can be achieved.

A point needs to be made about Plankton's technique of storing and reusing outcomes. When a PEC needs the relevant converged states of past PECs for its exploration, Plankton finds them by matching the non-deterministic choices made in the current run with those made for the past PECs. While we have not come across
a scenario where this is causes invalid combinations of converged states, if a configuration is created that breaks this assumption, Plankton could report false violations.

\subsection{Explicit-state model checker}
The explicit state model checker SPIN~\cite{cite:spin} provides Plankton its exhaustive exploration ability. SPIN verifies models written in the Promela modeling language, which has constructs to describe possible non-deterministic behavior. SPIN's exploration is essentially a depth-first search over the possible states of the supplied model. Points in the execution where non-deterministic choices are made represent branching in the state space graph. Plankton's network model is essentially an implementation of the control plane in Promela. Our current implementation supports OSPF, BGP and static routing. Recall from \S~\ref{subsec:pec} that Plankton partitions the header space into Packet Equivalence Classes. For each PEC, Plankton uses SPIN to exhaustively explore control plane behavior. In order to improve scalability, Plankton also performs another round of partitioning. Just as in the real world, Plankton executes the control plane for each prefix in the PEC separately. This separation of prefixes is helpful in simplifying the protocol model. However, it does limit Plankton's support for route aggregation. While features such as a change in the routing metric can be supported, if there is a route map that performs an exact match on the aggregated prefix, it will not apply to the more specific routes, which Plankton models. Once the converged states of all relevant prefixes are computed, a model of the FIB combines the results from the various prefixes and protocols into a single network-wide data plane for the PEC.

In order to reduce the model checker's search space, Plankton attempts to minimize irrelevant non-determinism, making the execution of the deterministic fragments of control plane execution comparable to simulation. Nevertheless, in this section, we begin the presentation of Plankton's network model with a simple, unoptimized view, having significant non-determinism. In \S~\ref{sec:optimizations}, we discuss how much of this non-determinism can be eliminated.

 \vspace{-0.4cm}
\subsection {Abstract Protocol Model}
To define a control plane suitable for modeling real world protocols such as BGP and OSPF, we look to the technique used by Minesweeper wherein the protocols were modeled as instances of the
stable paths problem. Along similar lines, we consider the Simple Path Vector Protocol~\cite{cite:stablePaths}, which was originally proposed to solve the stable paths problem. We first extend SPVP to support some additional features that we wish to model. Based on this, we construct a protocol we call the {\it Reduced Path Vector Protocol}, which we show to be sufficient to correctly perform model checking, if we are interested only in the converged states of the network. We use RPVP as the common control plane protocol for Plankton. We begin with an overview of SPVP, highlighting our extensions.

\newcommand{\ribinn}{\operatorname{rib-in}_n}
\newcommand{\ribin}{\operatorname{rib-in}}
\newcommand{\peers}{\operatorname{peers}}
\newcommand{\bestpath}{\operatorname{best-path}}
\newcommand{\importf}[2]{\operatorname{import}_{#1,#2}}
\newcommand{\exportf}[2]{\operatorname{export}_{#1,#2}}

\vspace{-0.4cm}
\subsubsection{Extended SPVP}
\label{sec:spvp}
In extended SPVP, for each node $n$, and for each peer $n' \in \peers(n)$, $\ribinn(n')$ keeps track of the most recent advertisement of $n'$ to $n$. In addition, $\bestpath(n)$ keeps the best path that $n$ has to one of the {\it multiple} origins. The peers are connected using a reliable FIFO message buffer to exchange advertisements. Each advertisement consists of a path from the sender of the advertisement to an origin. In each step (which we assume is performed atomically) a node $n$ takes an advertisement $p$ from the buffer connected to peer $n'$, and applies an import filter on it ($\importf{n}{n'}(p)$). $n$ then updates $\ribinn(n')$ with the new imported advertisement. In our extension, we assume that each node has a ranking function $\lambda$ that provides a {\it partial order} over the paths acceptable by the node. $n$ then proceeds by updating its $\bestpath$ to the highest ranking path in $\ribinn$. If the best path in $\ribinn$ have the same rank as the current best path and that path is still valid, $\bestpath(n)$ will not change. If the best path is updated, $n$ advertises the path to its peers. For each peer $n'$, $n$ applies the export filter on the path ($\exportf{n}{n'}(\bestpath(n))$) before sending the advertisement.

The import filter, the export filter, and the ranking functions are abstract notions that will be inferred from the
configuration of the node. However, we do make reasonable assumptions about these notions (Appendix~\ref{appendix:assumptions}).
Attributes such as communities, MED, IGP cost, etc. are accounted for in the the ranking function and the import/export filters.

If a session between two peers fails, the messages in the buffer are lost and the buffer
cannot be used anymore. We assume that when this happens, each peer gets $\bot$ as the advertised path.
Additionally, to be able to model iBGP, in extended SPVP we allow the ranking function of any node $n$ to change at any time during the execution of the protocol. This is to model cases in which for example a link failure causes IGP costs to change. In such cases we assume that $n$ receives a special message to recompute its $\bestpath$ according to the new ranking function.

The state of network at each point in time consists of the values of $\bestpath$, $\ribin$, and the contents of the message buffers.
In the initial state $S_0$, the best path of the origins is $\eps$ and the best path of the rest of the nodes are $\bot$ which
indicates that the node has no path.
Also for any $n,n' \in V$, $\ribinn(n')$ is $\bot$.
We assume that initially the origins put their advertisements in the message buffer to
their peers, but the rest of the buffers are empty.

An (partial) execution of SPVP is a sequence of states $\pi$ which starts from $S_0$ and each state is reachable by a single atomic step of SPVP on the state before it. A converged state in SPVP, is a state in which all buffers are empty. A complete execution is an execution that ends in a converged state. It is well known that there are configurations which can make SPVP diverge in some or all execution paths. However, our goal is to only to check the forwarding behavior in the converged states, through explicit-state model checking. So, we define a much simpler model that can be used, without compromising the soundness or completeness the analysis (compared to SPVP).

\subsubsection {Reduced Path Vector Protocol}
\label{subsec:rpvp}

We now describe RPVP, which is specifically designed for explicit-state model checking of the converged states of the extended SPVP protocol.

\newcommand{\invalid}{\operatorname{invalid}}

\begin{algorithm}[t]
\scriptsize

\begin{algorithmic}[1]

\Procedure {RPVP}:

  \State Init : $\forall n \in N - Origins . \bestpath(n) \gets \bot$
  \State Init : $\forall o \in Origins . \bestpath(o) = \eps$
  \While {true}:
     \State $E \gets \{n \in N | \invalid(n) \lor \exists n' \in \peers(n) . \operatorname{can-update_n}(n') \}$
     \If{$E = \emptyset$}:
       \State break
     \EndIf

     \State $n \gets \operatorname{nondet-pick}(E)$
     \If{$\invalid(n)$}
        \State $\bestpath(n) \gets \bot$
     \EndIf
     \State $\mathit{U} \gets \operatorname{best}(\{n' \in \peers(n) | \operatorname{can-update_n}(n')\})$\label{line:bestpeer}
     \State $n' \gets \operatorname{nondet-pick}(\mathit{U})$
     \State $p \gets \importf{n}{n'}(\exportf{n'}{n}(\bestpath(n')))$
     \State $\bestpath(n) \gets p $
 \EndWhile
\EndProcedure
 \end{algorithmic}
 \caption{RPVP}
 \label{spvp}
 \end{algorithm}

In RPVP, the message passing model of SPVP is replaced with a shared memory model, without $\ribin$ or
buffers. The network state only consists of the values of $\bestpath$ of each node at each moment.
In the initial state, the best path of all nodes is $\bot$, except origins, whose best path is $\eps$.
At each step, the set of all enabled nodes ($E$) is determined (line 5). A node $n$ is considered
enabled if either i) the current best path $p$ of $n$ is invalid, meaning next hop determined by the path
has a a best path that is not a prefix of $p$.

\scalebox{0.8}{
$\invalid(n) ::= \bestpath(\bestpath(n).\operatorname{head})
                \neq \bestpath(n).\operatorname{rest}
$
}

Or ii) there is a node $n'$ among the peers of $n$ that can produce an advertisement
that can update the current best path of $n$. In other words, $n'$ has a path better than the
current best path of $n$, and the path is acceptable according to the export and import policies
of $n'$ and $n$ respectively.

\scalebox{0.8}{
  $\operatorname{can-update_n}(n') ::=
 \operatorname{better}(\operatorname{import_{n,n'}}(\operatorname{export_{n',n}}(\operatorname{best}(n')), \operatorname{best}(n))$
 }

Where $better_n(p,p')$ is true when path $p$ is preferred over $p'$ according to the
ranking function of $n$.

If there is no enabled node at a step, RPVP has reached a converged state.
Otherwise a node $n$ is non-deterministically picked among the enabled set (line 9).
If the current best path of $n$ is invalid, the best path is set to $\bot$.
Among all peers of $n$ that can produce advertisements that can update the best path of $n$,
the neighbors that produce the highest ranking advertisements are selected (line 13).
Note that in our model we allow different paths to have the same rank,
so there may be more than one elements in the set $U$.
Among the updates, one peer $n'$ is non-deterministically selected and
the best path of $n$ is updated according to the advertisement of $n'$ (lines 14-16).
By the end of line 16, an iteration of RPVP is finished. Note that there are no
explicit advertisements propagated; instead nodes are polled for the advertisement
that they would generate based on their current best path when needed.
When the protocol finishes, we have a converged state for the prefix at hand.
RPVP does not define the semantics of failure or ranking function change.
Any topology changes to be verified are made before the protocol begins execution and
the ranking functions are in their latest version.

A natural question though, is whether or not performing analysis using
RPVP is sound and complete with respect to SPVP.
If we are only concerned about the converged states, the answer is positive:

\newcommand{\converged}{\operatorname{converged}}

\begin{theorem}
For any converged state reachable from the initial state of the network with a
particular set of links $L$ failing at certain steps during the execution of
SPVP, there is an execution of RPVP with the same import/export filters and
ranking functions equal to the latest version of ranking functions in the
execution of SPVP, which starts from the initial state in which all links in $L$
have failed before the protocol execution starts, and reaches the same converged
state. Particularly, there is a such execution in which at each step, each node
takes a best path that does not change during the rest of the execution.
\label{th:osvp}
\end{theorem}
\begin{proof}
  The proof can be found in the Appendix. In essence, we show a valid execution of RPVP in which the nodes
  converge in the same order and to the same best path as the execution of SPVP.
\end{proof}

Theorem~\ref{th:osvp} implies that performing model checking using the RPVP model is sound and complete. However, as optimized as RPVP is, it still has too much non-determinism to be used
directly in a model checker. We rely on domain-specific optimizations to eliminate much of this non-determinism and make model checking practical. We shall elaborate on these optimizations in \S~\ref{sec:optimizations}.

Our presentation of RPVP has assumed the that a single best path is picked by each node. This matches our current implementation in that we do not support multipath in all protocols. In a special-case deviation from RPVP, our implementation allows a node running OSPF to maintain multiple best paths, chosen based on multiple neighbors. While we could extend our protocol abstraction to allow multiple best paths at each node, it wouldn't reflect the real-world behavior of BGP which, even when multipath is configured, makes routing decisions based on a single best path. However, such an extension is valid under the constrained filtering and ranking techniques of shortest path routing. Theorem~\ref{th:osvp} and Theorem~\ref{th:det} (which we present in \S~\ref{subsubsec:det}) can be extended to incorporate multipath in such protocols. We do not do that to preserve clarity.

 \subsection{Policies}
 \label{subsec:policies}
 As we mention in previous sections, our design of Plankton targets verification of data plane policies over converged states of the network. Policies define an assertion check callback which will be invoked each time the model checker generates a converged state. The check is given access to the newly-computed converged data plane, as well as the relevant converged states of any other PEC that the current PEC depends on. The return value of the procedure is checked by the model checker, and if the policy has failed, it writes a trail file describing the path taken to reach the particular converged state. The Policy API allows policies to optionally supply two forms of additional information that would help to optimize Plankton's search: source nodes, which are the set of nodes where packets are assumed to originate for the assertion check function (for example, the starting points of graph traversal in reachability checks), and interesting nodes, which are special in some form in the context of the policy (for example, the source, destination and waypoints in a waypoint policy). We discuss how this information is used for optimizations in \S~\ref{subsec:policy-opt} and \S~\ref{subsec:failures}.

 Plankton's programmatic approach to policies is borrowed from VeriFlow. Just like VeriFlow, any policy that can be evaluated as a Boolean predicate on individual data plane states can be supported by Plankton. Combined with a traffic matrix, policies such as load balancing can also be checked. Implementing new policies does not require any knowledge of how Plankton generates the converged data plane states, although providing the optional optimization information will improve the overall performance. Currently, we have used this API to implement Reachability, Waypointing, Loop Freedom, BlackHole Freedom, Bounded Path Length, Path Consistency~\cite{cite:minesweeper} and Multipath Consistency~\cite{cite:minesweeper}.
 \vspace{-0.3cm}
\section {Optimizations}
\label{sec:optimizations}
Although Plankton's RPVP-based control plane greatly reduces the state space, naive model checking is still not efficient enough to scale to large networks. We address this challenge through
optimizations that fall into two major categories --- reducing the search space of the model checker, and making the search more efficient.

\vspace{-0.3cm}
\subsection {Explore consistent executions only}
\label{subsec:consistent}
To describe this optimization, we first introduce the notion of a \textit{consistent} execution:
For a converged state $S$ and a partial execution of RPVP $\pi$ , we say that $\pi$ is consistent with $S$ iff at each step of the execution, a node picks a path, and that path is present its path selection in state $S$.

Readers may notice that Theorem~\ref{th:osvp} asserts the existence of a consistent execution leading to each converged state of the network, once any failures have happened. This implies that if the model checker was to explore only executions that are consistent with some converged state, completeness of the search would not be compromised (Soundness is not violated since every such execution is a valid execution). Of course, when we start the exploration, we cannot know the exact executions that consistent with some converged state, and hence need to be checked. So, we simply assume that every execution we explore is relevant, and if we get evidence to the contrary (like a device having to change a selected best path), we stop exploring that execution.

\vspace{-0.2cm}
\subsection {Partial Order Reduction}
A well-known optimization technique in explicit-state model checking, Partial Order Reduction (POR) involves exploring a set of events only in one order, if the result of applying all of them will result in the same outcome, irrespective of the order. In general, precisely answering whether the order of execution affects the outcome can be as hard as model checking itself. In practice, model checkers such as SPIN provide conservative heuristics to achieve some reduction. However, in our experiments, this feature did not yield any reduction. We believe this is because our model of the network has only a single process, and SPIN's heuristics are designed to apply only in a multiprocess environment. Even if we could restructure the model to use SPIN's heuristics, we do not expect significant reductions, as evidenced in past work~\cite{cite:nice}. Instead, we implement POR heuristics, based on our knowledge of the RPVP control plane\footnote{Since we wish to check all converged states of the network, it can be argued that any reduction in search space is essentially POR. But here, we are referring optimizations that have a localized scope.}.

\vspace{-0.2cm}
\subsubsection {Deterministic nodes}
\label{subsubsec:det}
In typical networks, even when there is the possibility of non-deterministic convergence, the non-determinism is localized. In other words, after each non-deterministic choice is made, there exists a significant amount of deterministic behavior before the opportunity for another non-deterministic choice, if any, arises (we consider it analogous to endgames in chess, but applicable at any point in the protocol execution). We leverage this property to reduce the search space for Plankton's model checker.

For an enabled node $n$ in state $S$ with a single best update, we say $n$ is deterministic if in all possible converged states reachable from $S$, $n$ will have the path that will be selected after $n$ processes the update. At each step of RPVP, after finding the set of enabled nodes, if there is a non-empty set of deterministic nodes among the enabled nodes, we arbitrarily (not non-deterministically) choose one of these nodes and process its update. The correctness of doing this is established by the following theorem:

\begin{theorem}
Any partial execution of RPVP consistent with a converged state can be extended to a complete execution consistent with that state.
\label{th:det}
\end{theorem}
\begin{proof}
The proof can be found in the Appendix.
\end{proof}

It can be seen that in any given state, choosing a deterministic node to execute represents a step that is consistent with every possible converged state reachable from the current state of the system. Based on Theorem~\ref{th:det}, this step does not eliminate any of those converged states from being reachable, and hence, can always be taken without violating completeness. It should be noted that these optimizations do not require the entire network to have one deterministic converged state. The optimization is applicable at every step of the protocol execution, possibly between non-deterministic choices.

\vspace{-0.2cm}
\subsubsection{Identifying deterministic nodes}
Key to the effectiveness of prioritizing deterministic nodes is the ability to identify deterministic nodes in a given state. The more robust the detection of determinism, the better will be the reduction in the model checker's search space. However, since the detection algorithm runs at every step, it is also important to keep the algorithm quick. Our detection algorithm for OSPF runs a network wide shortest path computation, and picks each node only after all nodes with shorter paths have executed. Moreover, the shortest path computation is amortized over groups of states with the same shortest path tree.

For BGP, the detection algorithm performs the following computation:
For each device which is enabled to update its best path, it checks whether there exists a pending update that would never get replaced. The algorithm determines this by proceeding sequentially through the steps of the BGP decision process, checking whether a pending update would be preferred by the device over other updates that are enabled now, or may be enabled in future. If there is the possibility of receiving a more preferred update, the device is marked as not deterministic. If the update is one among multiple equally preferred updates, their preference in the next step of the decision process is compared. For each step of BGP decision process, the preference calculation is quite conservative. For local pref, it marks an update as the winner if it matches an import filter that explicitly gives it a higher local pref than other updates. For AS Path, the length of the AS path of the current update is compared to the shortest AS path lengths that may reach the device. Similarly, in each step, the preference for currently active updates is compared to others that it may receive. If at any node, any update is found to be a clear winner, the node is picked as a deterministic node, and is allowed to process the update. If no node is found that has a clear winner, but a node has a set of updates, each of can be the winner, a weaker form of this optimization picks the node, but the decision on which update to process is made by SPIN. Figure~\ref{fig:det} illustrates these scenarios on a BGP network, highlighting one (out of many) possible sequence of node selections in each step.

\begin{figure}
\begin{center}
\includegraphics[width=0.7\columnwidth, keepaspectratio]{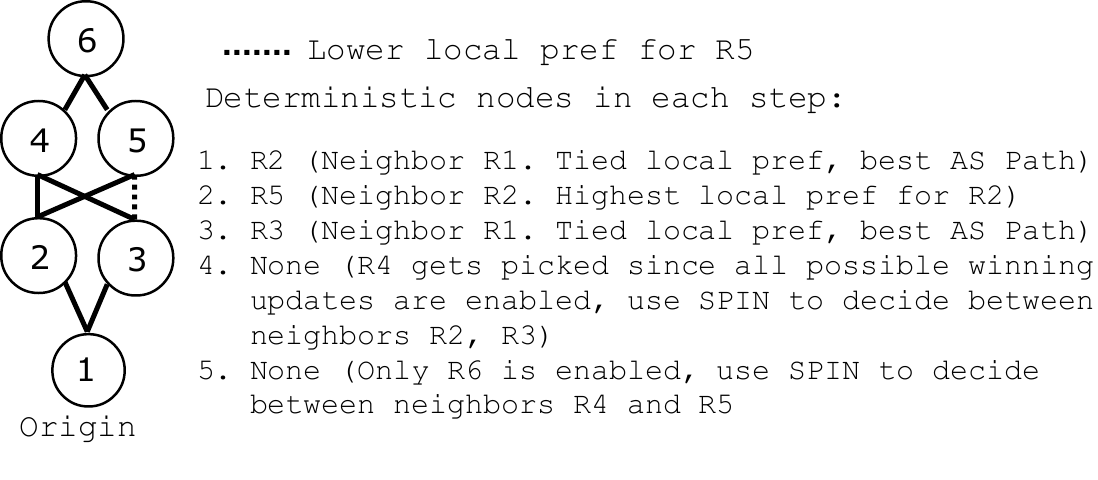}
\end{center}
\vspace{-0.8cm}
\caption{Step-by-step choice of deterministic nodes (Each node has a different AS number). }
\vspace{-0.2cm}
\label{fig:det}
\end{figure}

There are cases where the detection algorithm may fail to detect valid deterministic nodes. For instance, the algorithm for BGP may fail to mark a node $N$ as deterministic if it sets a high local pref for updates with a particular community attribute, but $N$ can never receive such an update from the neighbor for which that import filter is configured. The algorithms would also give up in cases where a prefix is redistributed from one routing process to another, and then back to original process. Nevertheless, since POR is only a heuristic, it is sufficient for these optimizations to apply reasonably frequently. We believe this will be the case. As long as at least one deterministic node can be found in each step, the execution will not have to fall back to the model checker.

We also note that even if the decision on a node is ambiguous in a particular state, the system will often make progress to a state where ambiguities can be resolved. For example, consider the case of the ``unused'' import filter that we just described. In a particular state, the algorithm may believe that a high priority update can arrive at node $N$ from the neighbor for which the import filter is defined. However, once the neighbor selected a path, the detection algorithm does not need to account for a more preferred update from it (recall from \S~\ref{subsec:consistent} that best paths never change in the executions that we explore).

\vspace{-0.2cm}
\subsubsection {Influence relation}
This optimization performs partial order reduction between the execution of two nodes $A$ and $B$ if the decision at $A$ cannot influence the decision at $B$ and vice versa. As with finding deterministic nodes, computing the influence perfectly is hard. So, if messages from $A$ can reach $B$ without passing through any node that has already made its best path decision, we assume that $A$ can influence $B$.

\vspace{-0.2cm}
\subsection {\vspace{-0.2cm}Policy-based Pruning}
\label{subsec:policy-opt}
This is the only optimization in Plankton that prunes the number of converged states that get checked as part of the verification. Policy-based Pruning limits protocol execution to those parts of the network that are relevant to the satisfaction/failure of the policy. When a policy defines a set of source nodes (\S~\ref{subsec:policies}), it indicates that the policy can be checked by analyzing the forwarding of the Packet Equivalence Class from those nodes only. The best example for this is reachability, which is checked from a particular set of starting nodes.

When an execution reaches a state where all source nodes have made their best-path decision, Plankton considers the execution, which is assumed to be consistent, to have finished.
In the cases where only a single prefix is defined in a PEC, Plankton performs a more aggressive pruning, based on the influence relation. Any device that cannot influence a source node is not allowed to execute. With some additional bookkeeping, the optimization can be extended to cases where multiple prefixes contribute to a PEC, but our current implementation does not support this. The optimization is also not sound when applied to PECs on which other PECs depend. For example, a router that does not speak BGP may not directly influence a source node, but it may influence the routing for the router IP addresses, which in turn may affect the chosen path of the source node. So, the optimization is not applied in such cases.
\vspace{-0.2cm}
\subsection {Failures}
\label{subsec:failures}
Plankton uses two kinds of optimizations in its selection of link failures to explore. The first is Partial Order Reduction.
As stated in \S~\ref{subsec:rpvp}, the model checker performs all topology changes before the protocol starts execution. We also enforce a strict ordering of link failures, reducing branching even further. The second kind of optimization uses the equivalence partitioning of devices as done by Bonsai~\cite{cite:cp-compression}. Bonsai groups devices in the network into abstract nodes, creating a smaller topology overall for verification. Plankton uses similarly Device Equivalence Classes (DECs) to restrict the link failures to be explored. Only one link is picked among those that connect the same DECs. If more than one link needs to be picked, the partitioning is refined after each selection. Keep in mind that this is an optimization limited to choosing link failures, the verification happens on the original input network. In order to avoid remapping interesting nodes(\S~\ref{subsec:policies}), they are each assigned to a separate DEC. Our current implementation of this optimization supports only cases where there are no dependencies across Packet Equivalence Classes.

\subsection{State hashing}
During the exhaustive exploration of the state space, the explicit state model checker needs to track a large number of states simultaneously. A single network state consists of a copy of all the protocol-specific state variables at all the devices. Maintaining millions of copies of these variables is naturally expensive. We observe that across a large number of network states, many of these variables remain unchanged. For example, a routing decision at one device doesn't immediately affect the variables at any of the other devices. Plankton leverages this property to reduce the overall memory footprint. Plankton stores parts of the state variables in a duplicate-eliminated hash table, and stores a 64-bit pointer in the network state. In our experiments, this optimization improves the memory usage as much as 20$\times$ in some cases.
 \vspace{-0.25cm}
\section{Evaluation}
\label{sec:eval}

We implemented a prototype of Plankton including the equivalence class computation, control plane model, policy API and optimizations in 373 lines of Promela and 4870 lines of C++, excluding the SPIN model checker. We experimented with our prototype on Ubuntu 16.04 running on a 3.4 GHz Intel Xeon processor with 32 hardware threads, with 188 GB RAM.

We begin our evaluation with simple hand-created topologies incorporating protocol characteristics such as shortest path routing, non-deterministic protocol convergence, redistribution, recursive routing etc. Among these tests, we incorporated examples of non-deterministic protocol execution from~\cite{cite:stablePaths}, as well BGP wedgies, which are policy violations which can occur only under some non-deterministic execution paths. In each of these cases, Plankton correctly finds any violations that are present.

\begin{figure*}[ht!]
    \centering
    \begin{subfigure}{0.35\linewidth}
        \centering
        \hskip-1em
        \includegraphics[width=1.02\columnwidth, keepaspectratio]{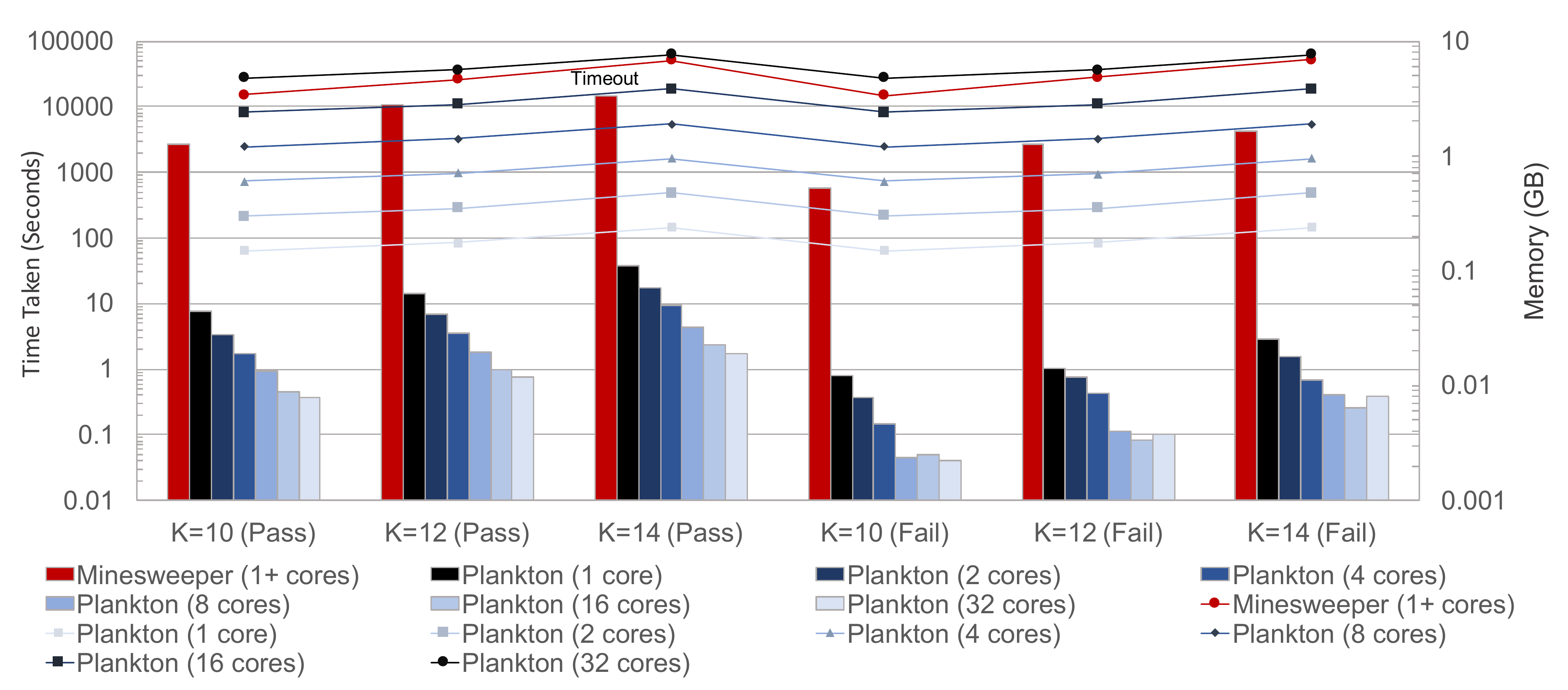}
        \vspace{-0.2cm}
        \caption{Loop policy - Fat Trees with OSPF}
    \end{subfigure}
    \hskip0.1em
    \begin{subfigure}{0.35\linewidth}
        \centering
        \includegraphics[width=1.02\columnwidth, keepaspectratio]{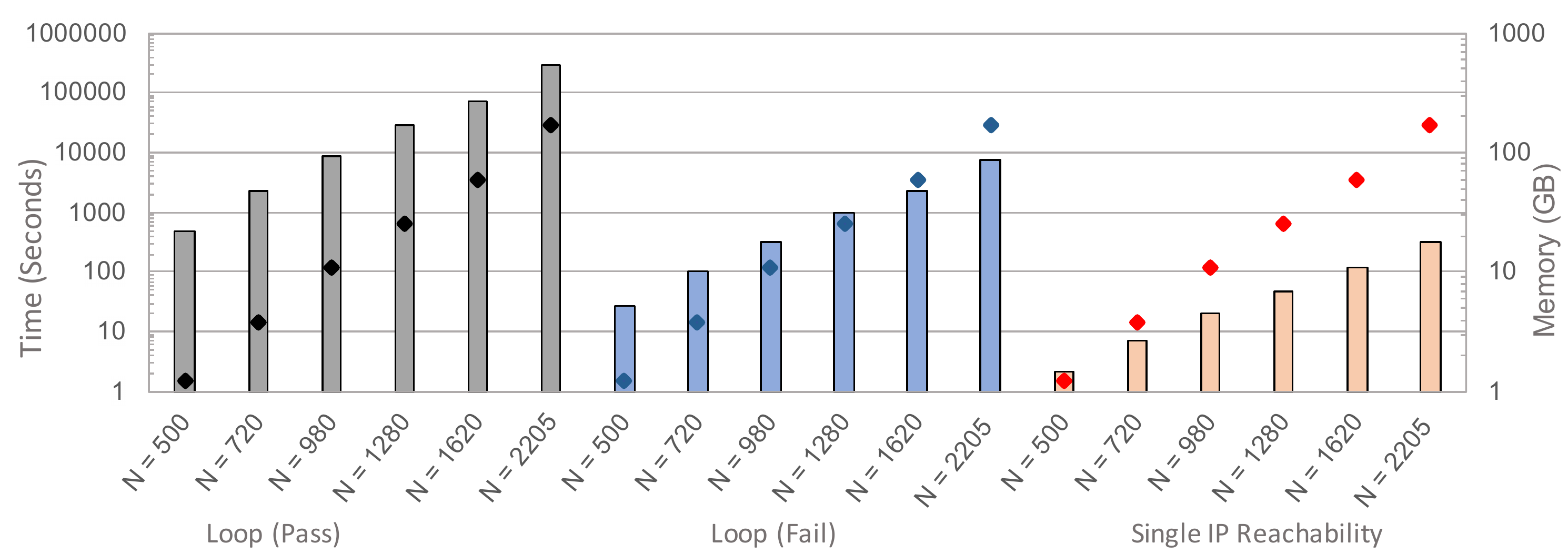}
        \caption{Fat trees running OSPF (1 core)}
    \end{subfigure}
    \hskip1em
    \begin{subfigure}{0.25\linewidth}
    \centering
    \includegraphics[width=1.02\columnwidth, keepaspectratio]{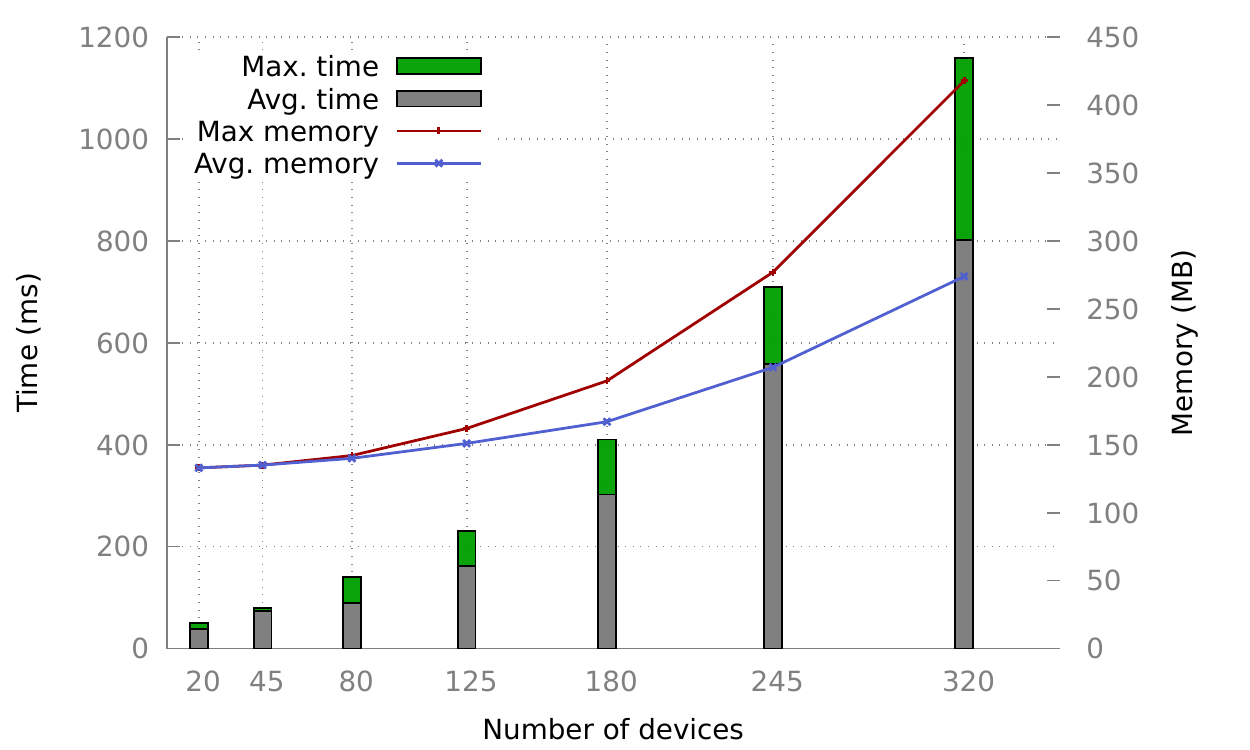}
    \vskip-0.5em
    \caption{BGP on Fat trees (Waypoint)}
    \vspace{-0.25cm}
    \end{subfigure}
    \begin{subfigure}{0.31\linewidth}
        \centering
        \includegraphics[width=1.02\columnwidth, keepaspectratio]{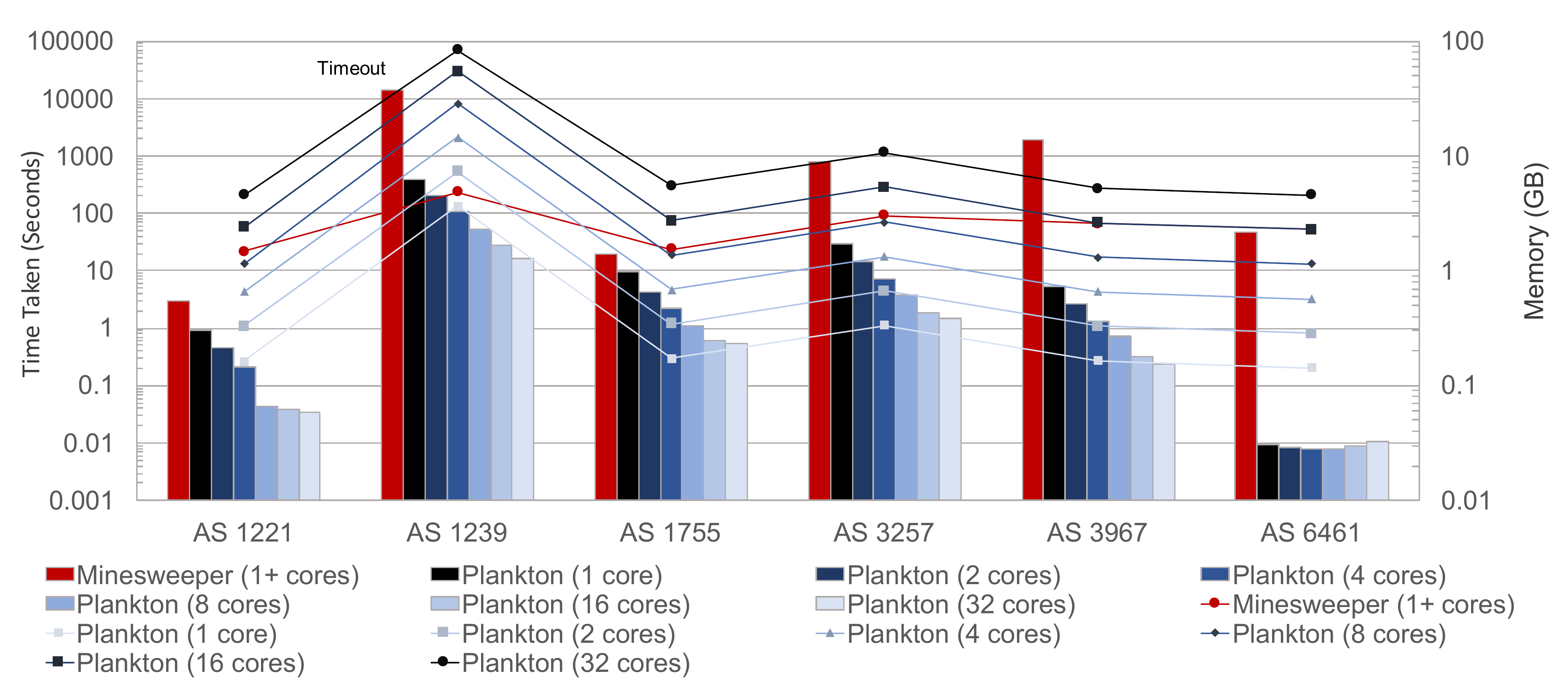}
        \caption{Reachability under failures in AS topologies.}
        \label{rocketfuel-fail}
    \end{subfigure}
    \hskip1em
    \begin{subfigure}{0.35\linewidth}
        \centering
        \vskip0.3em
        \includegraphics[width=1.02\columnwidth, keepaspectratio]{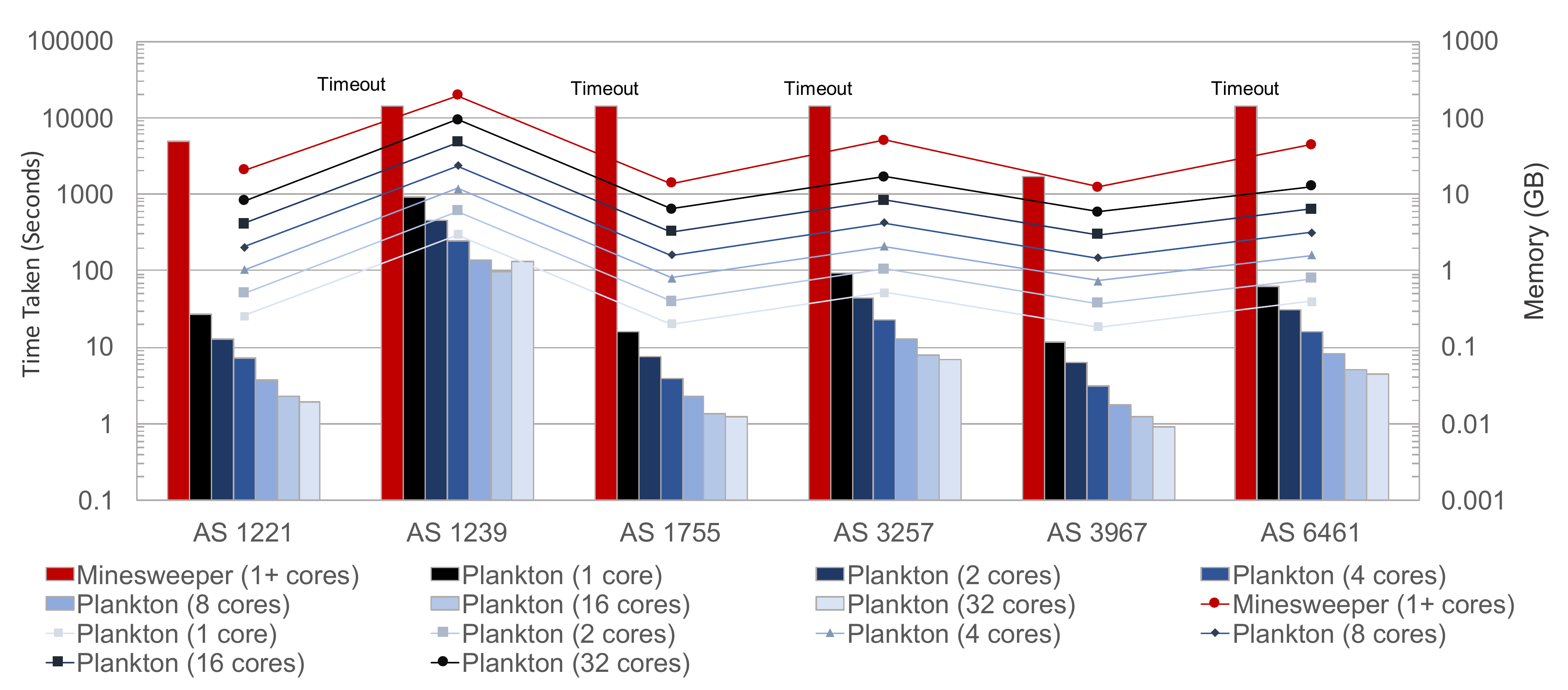}
        \caption{Reachability in iBGP over OSPF in AS topologies.}
    \end{subfigure}
    \hskip1em
    \begin{subfigure}{0.25\linewidth}
        \centering
        \vskip-2em
        \includegraphics[width=1.05\columnwidth]{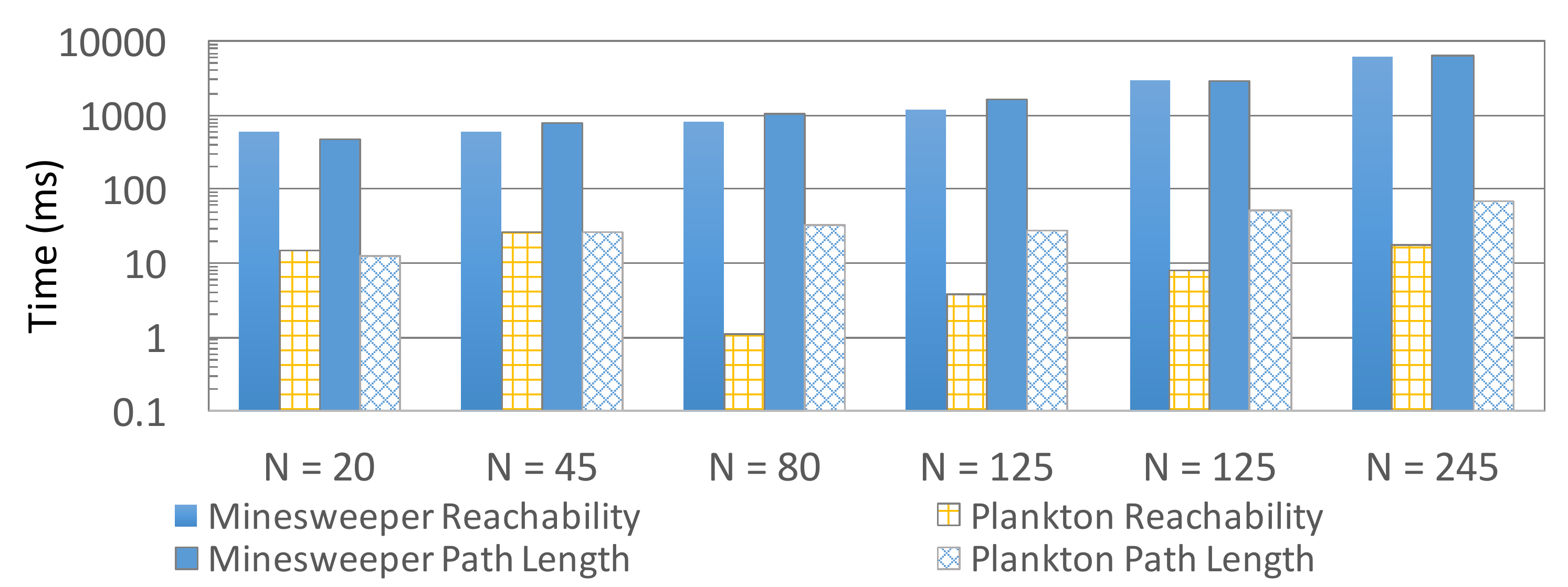}
        \caption{Bonsai-compressed fat trees.}
    \end{subfigure}
    \begin{subfigure}{0.3\linewidth}
        \centering
       \scalebox{0.5}{
        \begin{tabular}{ p{0.5\columnwidth} | p{0.2\columnwidth} |  p{0.3\columnwidth} | p{0.4\columnwidth} }
          \hline
          {\small Network} & {\small Links Failed} & {\small ARC} & {\small Plankton} \\
          \hline
          \hline
          \multirow{3}{*}{\small Fat tree (20 nodes) } & {\small 0} & {\small \hfill 1.08 s} & {\small\hfill 0.19 s} \\ 
          \cline{2-4}
          & {\small $\leq$ 1} & {\small \hfill 1.05 s} & {\small\hfill 0.23 s} \\  
          \cline{2-4}
          & {\small $\leq$ 2} & {\small \hfill 1.00 s} & {\small\hfill 0.31 s} \\ 
          \hline
          \multirow{3}{*}{\small Fat tree (45 nodes)} & {\small 0} & {\small \hfill 12.17 s} & {\small\hfill 0.49 s} \\ 
          \cline{2-4}
          & {\small $\leq$ 1} & {\small \hfill 12.40 s} & {\small\hfill 0.55 s} \\% 553.01 ms
          \cline{2-4}
          & {\small $\leq$ 2} & {\small \hfill 12.49 s} & {\small\hfill 2.09 s} \\
          \hline
          \multirow{3}{*}{\small Fat tree (80 nodes) } & {\small  0} &  {\small \hfill 300.50 s} & {\small\hfill 0.93 s} \\ 
          \cline{2-4}
          & {\small $\leq$ 1} &  {\small \hfill 280.00 s} & {\small\hfill 1.98 s} \\
          \cline{2-4}
          & {\small $\leq$ 2} &  {\small \hfill 294.12 s} & {\small\hfill 18.30 s} \\
          \hline
          \multirow{3}{*}{\small Fat tree (125 nodes)} & {\small  0} &  {\small \hfill 4847.57 s} & {\small\hfill 1.90 s} \\
          \cline{2-4}
          & {\small $\leq$ 1} &  {\small \hfill 5096.96 s} & {\small\hfill 9.34 s} \\
          \cline{2-4}
          & {\small $\leq$ 2} &  {\small \hfill 4955.95 s} & {\small\hfill 159.37 s} \\
          \hline
          \multirow{3}{*}{\small AS 1221 (108 nodes)} & {\small  0} &  {\small \hfill 41.62 s} & {\small\hfill 1.12 s} \\
          \cline{2-4}
          & {\small $\leq$ 1} &  {\small \hfill 40.50 s} & {\small\hfill 2.80 s} \\
          \cline{2-4}
          & {\small $\leq$ 2} &  {\small \hfill 38.83 s} & {\small\hfill 106.57 s} \\
          \hline
          \multirow{3}{*}{\small AS 1775 (87 nodes)} & {\small  0} &  {\small \hfill 49.32 s} & {\small\hfill 1.02 s} \\
          \cline{2-4}
          & {\small $\leq$ 1} &  {\small \hfill 48.65 s} & {\small\hfill 2.73 s} \\
          \cline{2-4}
          & {\small $\leq$ 2} &  {\small \hfill 46.52 s} & {\small\hfill 155.28 s} \\
          \hline
        \end{tabular}
        }
        \caption{All-to-all reachability policy in ARC vs. Plankton, with link failures}
    \end{subfigure}
    \hskip0.1em
    \begin{subfigure}{0.4\linewidth}
    \centering
    \vskip-2em
    \includegraphics[width=\columnwidth, keepaspectratio]{figs/real-world.pdf}
    \vskip-0.4em
    \caption{Real World Config Checks I - Number of devices in parentheses}
    \vskip-0.4em
    \end{subfigure}
    \hskip0.5em
    \begin{subfigure}{0.25\linewidth}
       \resizebox{\linewidth}{!}{
        \begin{tabular}{ p{0.35\columnwidth} | p{0.7\columnwidth} | p{0.3\columnwidth} |  p{0.4\columnwidth} | p{0.4\columnwidth} }
          \hline
          {\small Network} & {\small Policy} & {\small Links Failed} & {\small Memory} & {\small Time} \\
          \hline
          \hline
          \multirow{6}{*}{\small II} & \multirow{2}{*}{\small \hfill Loop} & {\hfill 0} & {\small \hfill 2.37 GB} & {\small\hfill 8.79 s} \\
          \cline{3-5}
          & &{\small \hfill $\leq 1$} & {\small \hfill 2.47 GB} & {\small\hfill 13.62 s} \\
          \cline{2-5}
          &\multirow{2}{*}{\small Multipath Consistency} & {\small \hfill 0} & {\small \hfill 2.37 GB} & {\small\hfill 16.28 s} \\
          \cline{3-5}
          && {\small \hfill $\leq 1$} & {\small \hfill 2.37 GB} & {\small\hfill 23.57 s} \\
          \cline{2-5}
          &\multirow{2}{*}{\small Path Consistency} & {\small \hfill 0} & {\small \hfill 2.37 GB} & {\small\hfill 16.26 s} \\
          \cline{3-5}
          && {\small \hfill $\leq 1$} & {\small \hfill 2.37 GB} & {\small\hfill 23.55 s} \\
          \hline
          \multirow{6}{*}{\small III} & \multirow{2}{*}{\small \hfill Loop} & {\hfill 0} & {\small \hfill 2.36 GB} & {\small\hfill 11.49 s} \\
          \cline{3-5}
          & &{\small \hfill $\leq 1$} & {\small \hfill 2.88 GB} & {\small\hfill 29.81 s} \\
          \cline{2-5}
          &\multirow{2}{*}{\small Multipath Consistency} & {\small \hfill 0} & {\small \hfill 2.37 GB} & {\small\hfill 16.33 s} \\
          \cline{3-5}
          && {\small \hfill $\leq 1$} & {\small \hfill 2.67 GB} & {\small\hfill 24.86 s} \\
          \cline{2-5}
          &\multirow{2}{*}{\small Path Consistency} & {\small \hfill 0} & {\small \hfill 2.36 GB} & {\small\hfill 16.40 s} \\
          \cline{3-5}
          && {\small \hfill $\leq 1$} & {\small \hfill 2.88 GB} & {\small\hfill 25.03 s} \\
          \hline
          \multirow{6}{*}{\small IV} & \multirow{2}{*}{\small \hfill Loop} & {\hfill 0} & {\small \hfill 2.31 GB} & {\small\hfill 12.37 s} \\
          \cline{3-5}
          & &{\small \hfill $\leq 1$} & {\small \hfill 2.40 GB} & {\small\hfill 13.14 s} \\
          \cline{2-5}
          &\multirow{2}{*}{\small Multipath Consistency} & {\small \hfill 0} & {\small \hfill 2.34 GB} & {\small\hfill 16.36 s} \\
          \cline{3-5}
          && {\small \hfill $\leq 1$} & {\small \hfill 2.37 GB} & {\small\hfill 17.04 s} \\
          \cline{2-5}
          &\multirow{2}{*}{\small Path Consistency} & {\small \hfill 0} & {\small \hfill 2.33 GB} & {\small\hfill 16.33 s} \\
          \cline{3-5}
          && {\small \hfill $\leq 1$} & {\small \hfill 2.37 GB} & {\small\hfill 17.00 s} \\
          \hline
        \end{tabular}
        }
        \vskip-0.2em
        \caption{More Real World Config Checks}
    \end{subfigure}
    \caption{Plankton experiments. Bars denote time and lines/points memory consumption}
    \label{fig:experiments}
\end{figure*}

Having verified basic correctness, we start evaluating Plankton's performance on larger networks, and in comparison to existing configuration verifiers. Recall that Plankton's design goals are two-pronged --- to have good coverage in terms of protocols and converged data plane states, and also to achieve scalable performance. As Figure~\ref{stateoftheart} shows, what sets Plankton from tools other than Minesweeper is its higher data plane coverage, and the ability to handle more than just specific  protocols. So, the most pertinent verifier to compare Plankton's scalability with is Minesweeper. Nevertheless, in addition to multiple experiments comparing Plankton and Minesweeper, we also perform experimental comparison of Plankton with ARC.

Bonsai deserves additional discussion in this context. Bonsai is a preprocessing technique that helps improve the scalability of configuration verification, for specific policies. Bonsai could assist any configuration verifier, as demonstrated by the authors through experiments combining Bonsai with both Minesweeper as well as Batfish. We have now integrated Bonsai with Plankton too. We experimentally compare the performance of Minesweeper and Plankton, when both of them are integrated with Bonsai. However, it should be noted that it is important to study the performance of these tools without Bonsai too --- Bonsai's reductions cannot be applied if the correctness is to be evaluated under link failures, or if the policy being evaluated is not preserved by Bonsai.

\vskip0.4em
{\noindent \bf Experiments with synthetic configurations}\\
The first set of performance tests we perform is with fat trees. We construct fat-trees of increasing sizes, with each edge switch originating a prefix into OSPF. The link weights are assigned uniformly by us. Our intention is to check these networks for routing loops. In order to cause violations, we install static routes at the core routers. In our first set of experiments, the static routes match the routes that OSPF would eventually compute, so there are no loops. Then, we change the static routes such that some of the traffic falls into a routing loop. Figure~~\ref{fig:experiments}.a illustrates the time and memory consumed in performing the loop policy, using Plankton running on various number of cores, and using Minesweeper. We observed that under default settings, Minesweeper's CPU utilization keeps changing, ranging from 100\% to 1600\%. In this experiment and all others where we run both Minesweeper and Plankton, the results obtained were consistent with each other. This serves as an additional correctness check for Plankton. Bonsai is not used, because the currently available implementation of Bonsai does not appear to support routing loop checks.

As the results show, Plankton scales well with input network size. The speed and memory consumption varies as expected with the degree of parallelism. Even on a single core, Plankton is quicker than Minesweeper for all topologies. For larger networks, Plankton is several orders of magnitude quicker. On the memory front, even on 16 cores, Plankton's footprint is smaller than Minesweeper's.

Encouraged by the good performance numbers, we scale up the input size to very large fat trees. In these experiments, we observed that Minesweeper doesn't finish even in 4 hours, even with a 500 device network (in the case of passing loop check, even in a 245 device network). So, we did not run Minesweeper on the even larger networks. We report only the numbers obtained by running Plankton with a single CPU core (Figure~\ref{fig:experiments}.b). This is to illustrate the time-memory tradeoff that Plankton provides --- since the analysis of individual PECs are fully independent and of identical computational effort, running the experiment with $n$ cores would reduce the runtime by a factor of $n$, but raise the memory footprint by the same factor. Checking routing loops on some of the larger networks with a high degree of parallelism would require an industrial grade server (8 hours running on 10 cores on an AWS \texttt{x1.32xlarge} instance with 2 TB of RAM). Policies that require checking a single equivalence class do not suffer from this problem. For example, as the figure shows, single-IP reachability check finishes in seconds or minutes even on the largest networks.

In order to test Plankton with a very high degree of non-determinism, we evaluated BGP in a data center setting, which is often employed to provide layer 3 routing down to the rack level in modern data centers~\cite{cite:bgprfc}. We configure BGP as described in RFC 7938~\cite{cite:bgprfc} on fat trees of various sizes. Furthermore, we suppose that the network operator intends that traffic should pass through any of a set of acceptable \emph{waypoints} on the aggregation layer switches (This may be because the switches implement certain middlebox functionality; we pick a random subset of aggregation switches as the waypoints in each experiment). However, we create a ``misconfiguration'' that prevents multipath from being used and fails to steer routes through the waypoints\footnote{Plankton cannot verify the ``correct'' version, since it does not support BGP multipath}. Thus, in this scenario, whether the selected path passes through a waypoint depends on the order in which updates are received at various nodes, due to age-based tiebreaking~\cite{cite:oldestpath}. We check waypoint policies which state that the path between two edge switches should pass through one of the waypoints. Plankton evaluates various nondeterministic convergence paths in the network, and determines a violating sequence of events. Figure~\ref{fig:experiments}.c illustrates the average and worst-case time taken by Plankton to finish the check. Time and memory both depend somewhat on the chosen set of aggregation switches, but even the worst-case times are less than 2 seconds. We consider this a success of our Policy-based Pruning technique, since the network has too many converged states to be verified in reasonable time.

\vskip0.4em
{\noindent \bf Experiments with semi-synthetic configurations}\\
The next set of experiments we perform is with real-world AS topologies and configurations obtained from RocketFuel~\cite{cite:rocketfuel}. For the first experiment, we pick a random ingress point that has more than one link incident on it. From this ingress point, we verify that no advertised prefix should be dropped under single link failure. This experiment is an example of the case where Minesweeper's SMT-based search algorithm should be beneficial, due to the large search space created by failures. In this experiment (as well the next one), we do not use Bonsai, for two reasons --- i. Bonsai cannot be used for checks involving link failures, and ii. The topology has hardly any symmetry that Bonsai could exploit. Figure~\ref{fig:experiments}.d shows the time and memory measurements. We observe that even in this experiment, Plankton performs consistently better than Minesweeper, both in terms of time as well as memory overhead. A violation is found in every case. The time taken by Plankton with 16 and 32 cores are often identical, since a violation is found in the first set of PECs.

To evaluate our handling of PEC dependencies, we configure iBGP over OSPF on the AS topologies.The iBGP prefixes rely on underlying OSPF routing process in order to reach the next hop. We check that the iBGP prefixes are correctly delivered to the originating device. It is worth noting that this test evaluates a feature that, to the best of our knowledge, is provided only by Plankton and Minesweeper.
Figure~\ref{fig:experiments}.e illustrates the results. Unsurprisingly, Plankton performs multiple orders of magnitude better, thanks to the dependency-aware scheduler. Due its approach of duplicating the network, Minesweeper is solving a {\it much} harder problem here, sometimes over 300 times larger.

\vskip0.4em
{\noindent \bf Integration with Bonsai}\\
As mentioned earlier, we have integrated Plankton with Bonsai, so as to take advantage of control plane compression when permitted by the specific verification task at hand. We test this integration experimentally by checking Bounded Path Length and Reachability policies on fat trees running OSPF. The symmetric nature of fat-trees is key for Bonsai to have a significant impact. We measure the time taken by Plankton and Minesweeper to finish verification, {\em after} Bonsai preprocesses the original input. As Figure~\ref{fig:experiments}.f shows,
Plankton still outperforms Minesweeper by multiple orders of magnitude.

\vskip0.4em
{\noindent \bf Comparison with ARC}\\
Having evaluated Plankton's performance in comparison with Minesweeper, we move on to comparing the performance of Plankton and ARC. ARC is a tool that is specifically designed to check shortest-path routing under failures, so we expected the performance to be much better than the more general-purpose Plankton, when checking networks compatible with ARC. We check all-to-all reachability in Fat Trees and AS topologies running OSPF, under a maximum of 0, 1 and 2 link failures. We observed that like Minesweeper, ARC's CPU utilization ranges from 100\% to 600\% under default settings. The numbers we report for Plankton are with 8 cores. As Figure\ref{fig:experiments}.g shows, Plankton is multiple orders of magnitude faster in most cases\footnote{The numbers we obtain for ARC are similar to those reported by the authors for similar sized networks. So we believe we have not misconfigured ARC.}. This is genuinely surprising. The only reason that {\em may} explain the observation is that ARC always computes separate models for each source-destination pair, whereas Plankton computes them based only the destination, when verifying destination address routing. Nevertheless, we do not believe that there is a fundamental limitation in ARC's design that would prevent it from outperforming Plankton on the networks that can be checked by either tool. Interestingly, while ARC's resiliency focused algorithm doesn't scale as easily as Plankton for larger networks (over 5 minutes for 80 devices), it actually performs {\em better} when the number of failures to be checked increases. Plankton on the other hand scales poorly when checking increasing levels of resiliency. We do not find this concerning, since most interesting checks in the real world involve only a small number of failures, and checking large sets of nondeterministic outcomes is not what Plankton is designed for. We also note that when we performed the experiments with Minesweeper, no check involving 2 failures ran to completion except the smallest fat tree.

\vskip0.4em
{\noindent \bf Testing with real configurations}\\
We used Plankton to verify 10 different real-world configurations 3 different organizations, including the publicly available Stanford dataset. We first check reachability, way pointing and bounded path length policies on these networks, with and without failures. All except one of these networks use some form of recursive routing, such as indirect static routes or iBGP. We feel that this highlights the significance of Plankton's and Minesweeper's support for such configurations. Moreover, the PEC dependency graph for these networks did not have any strongly connected components larger than a single PEC, which matches yet another of our expectations. Interestingly, we did find that the PEC dependency graph had {\em self loops}, with a static route pointing to a next hop IP within the prefix being matched. It is also noteworthy that in these experiments, the only non-determinism was in the choice of links that failed, which substantiates our argument that network configurations in the real world are largely deterministic. Figure~\ref{fig:experiments}.h illustrates the results, which indicate that Plankton can handle the complexity of real-world configuration verification.

In our next experiment with real world configs, we identify three networks where Loop, Multipath Consistency and Path Consistency policies are meaningful and non-trivial to check. We check these policies with and without link failures. Figure~\ref{fig:experiments}.i illustrates the results of this experiment. The results indicate that the breadth of Plankton's policies scale well on real world networks. The Batfish parser, which is used by Minesweeper, was incompatible with the configurations, so we could not check these configs on Minesweeper. However, the numbers we observe are multiple orders of magnitude better than those reported for Minesweeper on similar-sized real-world networks, for similar policies.

\vskip0.4em
{\noindent \bf Optimization Cost/Effectiveness}\\
To determine the effectiveness of Plankton's optimizations, we perform experiments with some optimizaitons disabled or limited. Figure~\ref{fig:no-opt} illustrates the results from these experiments. When all optimizations are turned off, naive model checking fails to scale beyond the most trivial of networks. The optimizations reduce the state space by 4.95 $\times$ in smaller networks and by as much 24968 $\times$ in larger ones.

\begin{figure}
\vspace{0.2cm}
\begin{center}
\resizebox{.6\linewidth}{!}{
\begin{tabular}{ p{0.5\columnwidth} | p{0.45\columnwidth}| p{0.15\columnwidth} | p{0.22\columnwidth} }
  \hline
  \multicolumn{2}{c|}{\small Experiment} & {\small Time} & {\small Memory} \\
  \hline
  \hline

\multirow{2}{*}{\small Ring, OSPF, 4 nodes, 1 failure} & {\small All optimizations off} & {\small \hfill 1.56 ms} & {\small \hfill 137.39 MB}\\
\cline{2-4}
& {\small All optimizations on} & {\small \hfill 343 $\mu$s} & {\small \hfill 137.43 MB}\\
\hline
\multirow{2}{*}{\small Ring, OSPF, 8 nodes, 1 Failure} & {\small All optimizations off} & {\small \hfill 0.13 s} & {\small \hfill 137.04 MB}\\
\cline{2-4}
& {\small All optimizations on} & {\small \hfill 623 $\mu$s} & {\small \hfill 143.22 MB}\\
\hline
\multirow{2}{*}{\small Ring, OSPF, 16 nodes, 1 Failure} & {\small All optimizations off} & {\small \hfill 266.48 s} & {\small \hfill 7615.57 MB}\\
\cline{2-4}
& {\small All optimizations on} & {\small \hfill 2.44 ms} & {\small \hfill 137.89 MB}\\
\hline
\multirow{2}{*}{\small Fat tree, OSPF, 20 nodes} & {\small All optimizations off} & {\small \hfill > 5 min} & {\small \hfill > 8983.55 MB}\\
\cline{2-4}
& {\small All optimizations on} & {\small \hfill 464 $\mu$s} & {\small \hfill 551.73 MB}\\
\hline
\multirow{2}{*}{\small  Fat tree, OSPF, 245 nodes} & {\small All optimizations on} &{\small \hfill 4.297 s} & {\small \hfill 1908 MB}\\
\cline{2-4}
 & {\small Link failure optimization off} &{\small \hfill 64.97 s} & {\small \hfill 72862 MB}\\
  \hline
  \multirow{2}{*}{\small AS 1221 iBGP} & {\small All optimizations on} & {\small  \hfill 27.54 s} & {\small  \hfill 254.22 MB}\\
  \cline{2-4}
  & {\small Deterministic node optimization off} & {\small  \hfill 25.43 s} & {\small  \hfill 254.34 MB}\\
  \hline
  \multirow{3}{*}{\small Fat tree, BGP, 20 nodes} & {\small All optimizations on} & {\small  \hfill 46 ms} & {\small  \hfill 137 MB}\\
  \cline{2-4}
  & {\small Deterministic node optimization off} & {\small \hfill > 5 min} & {\small  \hfill > 6144 MB}\\
  \cline{2-4}
  & {\small Policy-based pruning off} & {\small \hfill > 5 min} & {\small  \hfill > 6144 MB}\\

  \hline
\end{tabular}
}
\end{center}
\vspace{-0.2cm}
\caption{Experiments with optimizations disabled/limited}
\label{fig:no-opt}
\vspace{-0.5cm}
\vskip 0.1cm
\end{figure}

To evaluate device-equivalence based optimizations in picking failed links, we perform loop check on Fat Trees running OSPF
under single link failure, but with device equivalence optimization turned off.
This resulted in over a 15$\times$ reduction in speed, and a 38$\times$ increase in memory overhead, indicating the effectiveness of the optimization in networks with high symmetry.

In the next set of experiments, we measure the impact of our Partial Order Reduction technique of prioritizing deterministic nodes(\S~\ref{subsubsec:det}). We first try the iBGP reachability experiment with the AS 1221 topology, with the detection of deterministic nodes in BGP disabled. We notice that in this case the influence-based Partial Order Reduction produces reductions identical to the disabled optimization, keeping the overall performance unaffected. In fact, the the time improves by a small percentage, since there is no detection algorithm that runs at every step. We see similar results when we disable the optimization on the edge switches in our BGP data center example. However, this does not mean that the deterministic node detection can be discarded --- in the BGP data center example, when the optimization is disabled altogether, the performance falls dramatically. The next optimization that we study is Policy-based Pruning. On the BGP data center example, we attempt to check a waypoint policy, with policy-based optimizations turned off. The check times out, since it is forced to generate every converged data plane, not just the ones relevant to the policy. This is not surprising to us --- Plankton's primary use case is to verify largely deterministic networks, not those with very high degree of non-determinism like the BGP data center example.

\begin{figure}
\vspace{0.2cm}
\begin{center}
\scalebox{0.5}{
\begin{tabular}{ p{0.7\columnwidth} | p{0.3\columnwidth} | p{0.3\columnwidth} }
  \hline
  {\small Experiment} & {\small Bitstate Hashing} & {\small No Bitstate Hashing} \\
  \hline
  \hline
  {\small 180 Node BGP DC Waypoint (Worst Case)} & {\small \hfill 202 MB} & {\small\hfill 67 MB} \\
  \hline
  {\small 320 Node BGP DC Waypoint (Worst Case)} & {\small \hfill 428 MB} & {\small\hfill 215 MB} \\
  \hline
  {\small AS 1239 Fault Tolerance (2 cores)} & {\small \hfill 7.33 GB} & {\small \hfill 4.52 GB} \\
  \hline
  {\small AS 1221 Fault Tolerance (1 core)} & {\small \hfill 163.53 MB} & {\small\hfill 60 MB} \\
  \hline
\end{tabular}
}
\end{center}
\vspace{-0.4cm}
\caption{The effect of bitstate hashing on memory use}
\label{fig:bitstate}
\vspace{-0.6cm}
\end{figure}

SPIN provides a built-in optimization called
bitstate hashing, that uses a Bloom filter to keep track
of explored states, rather than storing them explicitly. This can cause some false negatives due to reduced coverage of execution paths. We find that bit state hashing provides significant reduction in memory in a variety of our test cases (Figure~\ref{fig:bitstate}). According to SPIN's statistics our coverage would be over 99.9\%. Nevertheless, we have not turned on bitstate hashing in our other experiments in favor of full correctness.
\vspace{-0.7cm}
 \section{Limitations}
\vspace{-0.4cm}
\label{sec:limitations}
Some of the limitations of Plankton, such as the lack of support for BGP multipath, limited support route aggregation have been mentioned in previous sections. As discussed in \S~\ref{subsec:scheduling}, Plankton may also produce false positives when checking networks with cross-PEC dependencies, because it expects that every converged state of a PEC may co-exist with every converged state of other PECs that depend on it. However, such false positives are unlikely to happen in practice, since real-world cases of cross-PEC dependencies (such as iBGP) usually involve only a single converged state for the recursive PECs. Our current implementation of Plankton assumes full visibility of the system to be verified, and that any dynamics will originate from inside the system. So, influences such as external advertisements need to be modeled using stubs that denote entities which originate them. A fundamental limitation of Plankton's technique is that it cannot detect issues caused by vendor-specific protocol implementations. This is a limitation that Plankton shares with all other formal configuration analysis tools in existence. However, SPIN allows the Promela model to be substituted with real code too. This is an option that may potentially be useful.
 \vspace{-0.35cm}
\section{Related Work}
\vspace{-0.15cm}
\label{sec:related}

{\bf Data plane verification:} The earliest offline network verification techniques,
such as Anteater~\cite{cite:anteater} and HSA~\cite{cite:hsa}, evolved to real-time tools such as VeriFlow~\cite{cite:veriflow},
NetPlumber~\cite{cite:netplumber} and DeltaNet~\cite{cite:deltanet}. They cannot verify configurations prior to deployment.

{\bf Configuration verification:} We discussed the state of the art of configuration verification in \S\ref{sec:motivation}. Plankton uses a novel combination of symbolic partitioning of header
space with explicit state model checking, to achieve significant improvement in scalability.

CrystalNet~\cite{cite:crystalnet} emulates of actual device VMs, and its results could be fed to a data plane verifier. However, this would not verify nondeterministic control plane dynamics. Simultaneously improving the fidelity of configuration verifiers in \emph{both} dimensions (capturing dynamics as in Plankton and implementation-specific behavior as in CrystalNet) appears to be a difficult open problem.

{\bf Optimizing network verification:} Libra~\cite{cite:libra} is a divide-and-conquer data plane verifier, which is related to our equivalence class-based partitioning of possible packets.

{\bf Model checking in the networking domain:} Past approaches that used model checking in the networking domain have focused almost exclusively on the network software itself,
either as SDN controllers, or protocol implementations~\cite{cite:nice, cite:modelcheckingsdn, cite:protocolimplementation}. Plankton uses model checking not to verify software, but to verify configurations.
 \vspace{-0.4cm}
\section{Conclusion and Future Work}
\vspace{-0.2cm}

We described Plankton, a formal network configuration verification tool that combines equivalence partitioning of the header space with explicit state model checking of protocol execution. Thanks to pragmatic optimizations such as Partial Order Reduction and state hashing, Plankton produces significant performance gains over the state of the art in configuration verification. Plankton's model-checking based design allows further improvements, such as checking transient states, incorporating real software etc. We plan to attempt them in future. Finally, further improvements to the Partial Order Reduction heuristics, possibly with guaranteed POR is yet another interesting avenue of future work.

{\footnotesize \bibliographystyle{acm}
\bibliography{main}}
\begin{appendices}
\section{Assumptions}
\label{appendix:assumptions}
We make the following assumptions in in our theoretical model of BGP.
\begin{packeditemize}
\item Both import/export filters return $\bot$ if the filter rejects the advertisement according to the configuration.
\item All import filters reject paths that cause forwarding loops. They also do not alter the path (unless the path is rejected).
\item All export filters for each node $n$ not rejecting an advertisement, will append $n$ at the end of the advertised path. No other modification is made to the path.
\item Path $\bot$ has the lowest ranking in all ranking functions.
\item The import/export filters never change during the execution of the protocol. Note that we do not make such assumption for the ranking functions.
\end{packeditemize}

Note that these are reasonable assumptions with repsect to how real world BGP works.

\section{Proof of Theorems}
\label{appendix:proofs}

\begin{proof}[Proof of Theorem~\ref{th:osvp}]
For a complete execution $\pi = S_0, S_1, ..., S_c$ of SPVP and a state $S_i$ in that execution, for any node $n$,
we say that $n$ is converged in $S_i$ for execution $\pi$ iff $n$ has already picked the path
it has in the converged state ($S_c$) and does not change it:
\begin{multline*}
  \converged_\pi(n, S_i) ::= \\
   \forall j . i \leq j \leq c : \bestpath_{S_j}(n) = \bestpath_{S_c}(n)
\end{multline*}
It it clear that when a node converges in an state, it remains converged (according to the definition above).

\begin{lemma}
  In any complete execution $\pi = S_0, S_1, ..., S_c$ of SPVP, for any sate $S_i$,
  for any two nodes $n$ and $n'$, if $\exists l : n' = \bestpath_{S_i}(n)[l]$
  \footnote{For a path $P = p_0, p_1, ..., p_n$, we denote $p_i$ as $P[i]$ and $p_i, p_{i+1}, ..., p_n$ as $P[i:]$}
  , and $\bestpath_{S_i}(n') \neq \bestpath_{S_i}(n)[l:]$, there is a $j$ ($i < j \leq c$)
  such that $\bestpath_{S_j}(n) \neq \bestpath_{S_i}(n)$.

\end{lemma}
\begin{proof}

This can be shown by a simple induction on the length of the prefix of the best path of $n$ from $n$ to up to $n'$ ($l$).
If $l = 1$ (i.e the two nodes are directly connected) then either $n'$ will advertise its path to $n$
and $n$ and will change its path
or the link between $n$ and $n'$ fails in which case $n$ will receive an advertisement with $\bot$ as the path (Section~\ref{sec:spvp}), which causes
$n$ to change its path. Note that the argument holds even if the ranking function of $n$ or $n'$ changes.
If $l > 1$, assuming the claim holds for lengths less than $l$, for $n''= \bestpath_{S_i}(n)[0]$, either $\bestpath_{S_i}(n'') \neq \bestpath_{S_i}(n)[1:]$
in which case due to induction hypothesis the claim holds, or $\bestpath_{S_i}(n'') = \bestpath_{S_i}(n)[1:]$
in which case we note that $n' = \bestpath_{S_i}(n'')[l-1]$, and since the length of the path $n''$ to $n'$ is less than $l$, by induction hypothesis
we know that eventually (i.e for a $j > i$) $\bestpath_{S_j}(n'') \neq \bestpath_{S_j}(n)[1:]$ and by induction hypothesis this will
lead to a change in the best path of $n$.

\end{proof}

\begin{corollary}
For any complete execution $\pi$ of SPVP , for any node $n$, any node along the best path of
$n$ in the converged state converges before $n$.
\label{col:bestpath}
\end{corollary}

Now consider a complete execution $\pi = S_0, S_1, ..., S_c$ of SPVP.
We will construct a complete execution of RPVP with $|N|$ steps resulting in the same converged state as $S_c$.
We start with a topology in which all the links that have failed during the execution of SPVP are already failed.
For any node $n$, we define $C_\pi(n) = min \{ i | \converged_\pi(n, S_i) \}$.
Consider the sequence $n_1, n_2, ..., n_{|N|}$ of all nodes sorted in the increasing order of $C_\pi$.
Now consider the execution of RPVP $\pi' = S'_0, S'_1, ..., S'_{|N|}$  which starts from the initial state of RPVP
and in each state $S'_i$, (a) either node $n_{i+1}$ is the picked enabled node and the node $p_i = \bestpath_{S_c}(n_{i+1})[0]$ is the picked best peer, or
(b) in case $p_i = \bot$, nothing happens and $S'_{i+1} = S'_i$.

First, note that (modulo the repeated states in case \textit{b}), $\pi'$ is a valid execution of
of RPVP: at each state $S'_i$ (in case \textit{a}), $n_{i+1}$ is indeed an enabled node since its best path at that state is $\bot$ and
according to corollary~\ref{col:bestpath}, $p_i$ has already picked its path.
Also $p_i$ will be in the set of best peers of $n_{i+1}$ (line~\ref{line:bestpeer} in RPVP).
Assume this is not the case, i.e there exists another peer $p'$ that can advertise a better path.
This means that in $S_c$ of SPVP, $p'$ can send an advertisement that is better
(according to the version of ranking functions in $S_c$) than the converged path
of $n_{i+1}$. This contradicts the fact that $S_c$ is a converged path.

Second, note that $S'_{|N|}$ is a converged state for RPVP, because otherwise, using similar reasoning
as above, $S_c$ can not be converged. Also it is easy to see that $\bestpath_{S'_{|N|}} = \bestpath_{S_c}$

Finally, note that in $\pi'$, once a node changes its best path from $\bot$, it does not
change its best path again.

\end{proof}

\begin{proof}[Proof of Theorem~\ref{th:det}]
We begin by making two observations about RPVP that are key to the proof:
\begin{packeditemize}
\item RPVP for a prefix can never converge to a state having looping paths.
\item If a node $u$ adopts the best path of a neighbor $v$, $v$ will be next hop of $u$.
\end{packeditemize}

Consider any converged state $S$. The theorem states that any partial execution that is consistent with $S$ can be extended to a full execution that leads to $S$.
We prove the theorem by induction on the length of the longest best path in $S$.

{\noindent \bf Base case: }If in a network a converged state exists where the best path at each node is of length 0, that means that each node is either an origin or doesn't have a
best path for the prefix. Since any execution apart from the empty execution (where no protocol event happens) is not consistent with this state, the theorem holds.

{\noindent \bf Induction hypothesis: }If a converged state exists in a network such that all best paths are of length $k$ or less, then any partial execution that is consistent with the
converged state can be extended to a full execution that reaches the converged state.

{\noindent \bf Induction step: }Consider a network with a converged state $S$ such that the longest best path is of length $k+1$. We first divide the nodes in the network into two classes --- $N$,
which are the nodes with best paths of length $k$ or less, and $N'$, which are nodes with best paths of length $k+1$. Consider a partial execution $\pi$ that is consistent with $S$.
We identify two possibilities for $\pi$:\\

{\it Case 1: }Every node that has executed in $\pi$ falls into $N$. In this case, we define a smaller network which is the subgraph of the original network, induced by the nodes in $N$.
In this network, the path selections made in $S$ will constitute a converged state. This is because in the original network, in the state $S$, the nodes in $N$ are not enabled to make state
changes. So, we can extend $\pi$ such that we get an execution $\pi'$ where nodes in $N$ match the path selections in $S$. Now, we further extend $\pi'$ with steps where each node in $N'$ reads
the best path from the node that is its nexthop in $S$ and updates its best path. When every node in $N$ has done this, the overall system state will reach $S$.\\

{\it Case 2: }At least one node in $N'$ has executed in $\pi$. In this case, we observe that since $\pi$ is consistent with $S$, by the definition of a consistent execution, no node
in the network has read the state of any node in $N'$. So, we can construct an execution $\pi'$ which has the same steps as $\pi$, except that any step taken by a node in $N'$ is skipped. As
in the previous case, $\pi'$ can be extended to reach a converged state in the subgraph induced by $N$. We extend $\pi$, first by using the steps that extend $\pi'$, and if necessary, taking
additional steps at nodes from $N'$ to reach $S$.
\end{proof}
\end{appendices}
 
\end{document}